\documentclass[3p]{elsarticle}

\usepackage{lineno}
\modulolinenumbers[5]

\journal{ }

\usepackage{enumitem}
\usepackage{amsthm}
\newtheorem*{remark}{Remark}
\theoremstyle{plain} \newtheorem{lemma}{Lemma}
\usepackage{amsmath}
\usepackage{amssymb}
\usepackage{natbib}
\usepackage[colorlinks,citecolor=blue,urlcolor=blue,filecolor=blue,backref=page]{hyperref}
\usepackage{graphicx}
\usepackage{multirow}
\usepackage{multicol}

\numberwithin{equation}{section}
\theoremstyle{plain} \newtheorem{thm}{Theorem}[section]
\newtheorem{cor}{Corollary}[section]

\begin{document}

\begin{frontmatter}

\title{Predictive Criteria for Prior Selection Using Shrinkage in Linear Models}
%\support{Support information of the article.}
%}
%\runtitle{Penalty Selection}

%\begin{aug}
\author{Dean Dustin \fnref{addr1}},
\author{ Bertrand Clarke\fnref{addr2}},
\author{ Jennifer Clarke\fnref{addr3}}

%\runauthor{D. Dustin et al.}

\address[addr1]{Department of Statistics, University of Nebraska-Lincoln,  340 Hardin Hall North, PO Box 830963, Lincoln, NE, 68583-0963.
   ddustin8@huskers.unl.edu % print email address of "e1"
  %  \printead{u1}
    %\printead*{e2}
}

\address[addr2]{Same;
    bclarke3@unl.edu
   % \printead{u2}
}

\address[addr3]{Same;
    jclarke3@unl.edu
   % \printead{u3}
}

\begin{abstract}

Choosing a shrinkage method can be done by selecting a penalty from a list of pre-specified
penalties or by constructing a penalty based on the data.   If a list of penalties
for a class of linear models is given, we provide comparisons based on sample size
and number of non-zero parameters under a predictive stability criterion based on
data perturbation.  These comparisons provide recommendations for penalty selection
in a variety of settings. If the preference is to construct a penalty customized for a given
problem, then we propose a technique based on genetic algorithms, again using
a predictive criterion.  We find that, in general, a custom penalty never performs
worse than any commonly used penalties but that there are cases the custom
penalty reduces to a recognizable penalty.  
Since penalty selection is
mathematically equivalent to prior selection, our method also constructs priors.

The techniques and recommendations we offer are intended for finite
sample cases.   In this context, we argue that predictive stability under perturbation is one
of the few relevant properties that can be invoked when the true model is not
known. Nevertheless, we study variable inclusion in simulations and,
as part of our shrinkage selection strategy, we include oracle property considerations.
In particular, we see that the oracle property typically holds for penalties that satisfy
basic regularity conditions and therefore is not restrictive enough to play a direct role
in penalty selection.  In addition, our real data example also includes
considerations merging from model mis-specification.

\end{abstract}

\begin{keyword}[class=MSC]
\texttt[Primary ]{62F15}
\texttt[; secondary ]{62C10, 60G25}
\end{keyword}

\end{frontmatter}

\section{Shrinkage and Prediction}
\label{intro}

 In the context of linear models, inference problems in which the number of parameters $p$
is bigger than the sample size $n$, i.e., with
$p>n$, are ill-posed and require some 
form of regularization to be solved.  That is, some
extra information must be added to ensure the existence of a reasonable solution.
The earliest form of this is called Tikhonov regularization and was used initially for
matrix inversion.   In Statistics,
ridge regression is probably the first occurrence of Tikhonov regularization,
see \cite{Hoerl:1962}.     By the early 1990's, $L^2$ regularization was common
in a neural networks context, see \cite{Sjoburg:Ljung:1992}, not just to ensure
that a solution existed but also to reduce variance.   An important step forward was replacing
the $L^2$ penalty with an $L^1$ penalty, see \cite{Tibshirani:1996}.  This shrinkage method is 
called the least absolute shrinkage and selection operator (LASSO).  It provides a form
of regularization that does variable selection as well as variance
reduction while ensuring solutions exist.  Another insight was the concept of 
an `oracle property' (OP)
first proved for a penalty called the smoothly clipped absolute deviation (SCAD) in the context
of linear models;
see \cite{Fan:Li:2001}.  The OP meant that asymptotically, as
$n \rightarrow \infty$, the parameter estimates from the SCAD penalty
behaved asymptotically as if the correct $p$ were known, i.e., the parameter estimates 
either went to zero or were consistent, asymptotically normal, and efficient according to
whether the variables they represented were not or were in the true model.

Over the last two decades, numerous shrinkage methods have been
proposed and studied as individual methods, representing individual priors
and penalties, for the purposes of parameter inference.   Here, we want to choose
a penalty or prior from a class of penalties or priors.  One way to do this is
to fix a list of shrinkage methods, find a basis for comparing them, and choose
the best.  Alternatively, we can take a `build-your-own' approach and
construct a shrinkage method from part of the data and use it on the rest of the data.
In this way we can choose the optimal shrinage method adaptively.  Here, we provide
both a comparison of `off the shelf' methods and a build-your-own technique (based
on genetic algorithms) under a predictive optimality criterion.  The build-your-own
shrinkage methods never perform worse that the off-the-shelf methods,
but may return an off-the shelf method as optimal.  We give an example
of each case.

More formally,  write a linear model (LM) of the form 
$Y = X\beta + \epsilon$ where $Y= (Y_1, \ldots , Y_n)^T$,
$X$ is an $n \times p$ matrix of design values, and $\epsilon = (\epsilon_1, \ldots, \epsilon_n)^T$
is random error $\epsilon_i \sim N(0, \sigma^2)$ for some $\sigma > 0$
and assume we have a data set ${\cal{D}} = \{(y_i, x_i) \mid i =1, \ldots , n\}$.
A nonadaptive shrinkage method gives parameter estimates 
\begin{eqnarray}
\hat{\beta} = \arg \min_{\lambda, \beta} \sum_{i=1}^n L_1(y_i - x_i^T\beta) + \lambda \sum_{j=1}^p L_2(\beta_j)
\label{nonadaptive}
\end{eqnarray}
where $L_1$ and $L_2$ are loss functions, $x_i$ is the $i$-th row of $X$,
and $\lambda$ is the decay parameter.  The term shrinkage arises from the
fact that as $\lambda \rightarrow \infty$,  each $\beta_j \rightarrow 0$.
Sometimes the first term on the right hand side of (\ref{nonadaptive})is replaced by a log-likelihood; in this case and if the likelihood is normal, $L_1$ corresponds to squared error.  If $L_2$ is also
squared error we get the optimization that gives ridge regression.

By contrast an adaptive shrinkage method typically gives estimates of the form
\begin{eqnarray}
\hat{\beta} = \arg \min_{\lambda, w^p, \beta} \sum_{i=1}^n L_1(y_i - x_i^T\beta) + \lambda \sum_{j=1}^p w_j L_2(\beta_j)
\label{adaptive}
\end{eqnarray}
where $w^p = (w_1, \ldots , w_p)^T$ and the $w_j$'s are weights on the individual $\beta_j$'s.
Again, the first term on the right hand side of (\ref{adaptive}) may be replaced by a log-likelihood.  Also, the dependence on $w_j$ in the
second term may be more complicated; we have represented the adaptivity of the constraint 
to the data as multiplicative in the penalty term 
for the sake of convenience.  
We regard the SCAD and minimax concave (MCV) penalties 
(\cite{Zhang:2010}) as 
adaptive because their penalty terms or constraints are data driven
even though they only introduce one extra non-multiplicative 
parameter (and have the OP).   This is not simply parameter counting. Note that the elastic net (EN), \cite{Zou:Hastie:2005}, introduces two
parameters and is not adaptive while the adaptive EN (AEN) is adaptive and has the OP;
see  \cite{Zou:Zhang:2009}.

As written, expression \eqref{adaptive} introduces $p$ new parameters, the $w_j$'s.   When $n$
is not large, it is unclear how to estimate the $w_j$'s.  Two main techniques have been
proposed.  One, due to \cite{Zou:2006} (see also \cite{Wang:etal:2007}), is to set $\hat{w}_j = 1/|\hat{\beta}_j|^\gamma$
where $\hat{\beta}_j$ is any $\sqrt{n}$-consistent estimator of $\beta_j$ (e.g., from SCAD
that only adds one extra parameter or from linear models when $n$ is large enough) 
and choose $\gamma$ by a cross-validation criterion.  Another method, due to \cite{Gian:Wang:2013}, sets
$\hat{w}_j = SE(\beta_{j, OLS})^\gamma/ |\hat{\beta}_{j, OLS}|^\gamma$,  using the ordinary
least squares (OLS) estimator.   This requires $n>p$ unless some
auxiliary technique is used to find the SE.  \cite{Gian:Wang:2013} argued that their
method works well in cases where there is high collinearity.  In practice, for both methods, $\gamma=1$ is used to avoid extra computation. Here, we have exclusively used
the Zou method since it does not require $n>p$ and has a nice interpretation:
As $\hat{\beta}_j \rightarrow 0$, $\hat{w}_j \rightarrow \infty$ forcing $\beta_j=0$ in
\eqref{adaptive}.

Formally, the OP is the following.  Write $\beta = (\beta_1, \beta_2)$ where $\beta_2$ represents the
zero components of $\beta$.  Under various
conditions, $\sqrt{n}$ consistent local minimizers 
$\hat{\beta} = (\hat{\beta}_1,  \hat{\beta}_2)$ from certain shrinkage criteria
(such as SCAD) satisfy
\begin{eqnarray}
\hat{\beta}_2 \rightarrow 0 \quad \hbox{and} \quad \sqrt{n}I(\beta_1, 0) (\hat{\beta}_1 - \beta_1) \rightarrow
N(0,1).
\nonumber
\end{eqnarray}
Roughly, shrinkage methods segregate into those that are
nonadaptive, i.e., introduce exactly one `decay' parameter and usually do not 
satisfy the OP, and those that introduce two or more decay parameters
and often satisfy the OP.  We will see that, contrary to initial impressions, 
the oracle property is not rare. 

We note that \eqref{nonadaptive} corresponds to a joint density $\rho$ on the
data and $\beta$.   Indeed, exponentiating gives that $\hat{\beta}$ corresponds to
the mode of the posterior
\begin{eqnarray}
\rho(\beta \mid Y^n) \propto e^{-\lambda \sum_{j=1}^p L_2(\beta_j)}
e^{- \sum_{i=1}^n L_1(y_i - x_i^T\beta) }
\label{jointdensity}
\nonumber
\end{eqnarray}
in which $L_2$ defines a prior on $\beta$ with hyperparameter $\lambda$.
A similar manipulation can be applied to \eqref{adaptive}.   This means that penalty
selection is mathematically equivalent to prior selection.
However, the equivalence is only mathematical because the class of reasonable
penalties is a proper subset of the class of reasonable priors.    In particular, most reasonable
penalties are convex -- our method in Sec. \ref{usingGAs} assumes convecity,
for instance -- but priors do not have to be log-convex.

Even though shrinkage methods were originally introduced as a way to solve the $n<p$ problem,
the OP uses $n \rightarrow \infty$.  This is partially ameliorated by results that
give analogs of the OP where $p$ increases much faster than $n$ does; see \cite{Fan:Lv:2013}.
However, the cost of amerlioration is often artificial conditions on the 
parameter space and/or design matrix.
Moreover, many of the original examples given to verify that shrinkage methods
were effective actually had $n > p$ and took $p=8$, see \cite{Fan:Li:2001}, \cite{Zou:2006},
and \cite{Wang:etal:2007}.   To the best of our knowledge
no systematic comparison of shrinkage methods for different relative sizes of $n$ and $p$
has been done. This paper seeks to fill this gap and provide recommendations
for when different methods work well.

Our comparison is in terms of predictive stability and accuracy of model selection. 
That is, after finding $\hat{\beta}$ for a given
shrinkage method we define the predictor
\begin{eqnarray}
\hat{Y}(x) = x^T \hat{\beta}
\label{predictorlm}
\end{eqnarray}
for $Y(x)$ at some new value $x$.    Then we evaluate how well $\hat{Y}$ predicts
when the data are perturbed.   We perturb the data using the technique of 
\cite{Luo:etal:2006}.  The idea is to add $N(0, \tau^2)$ noise to the $y_i$'s in
${\cal{D}}$ and then use part of the data to form a predictor and the rest of the data to 
evaluate the predictor.  We do this by generating `instability' curves, basically $L^2$ errors
as a function of $\tau$.   A good predictor will have instability curves with low values,  that 
are smooth, and increase slowly with $\tau$.  

We also use more conventional accuracy measures for variable selection. We argue that
 pairing the two provides an assessment that captures analogues of 
both variance (from the instability
curves) and bias (from the accuracy measures). We regard instability as 
more important because even if a variable
is incorrectly included its contribution may be small if its coefficient is near zero
and if it is incorrectly excluded the bias should show up in the instability curve.

A first contribution of this paper is the use of predictive stability to compare 
ten shrinkage methods for 
four values of $n$ when $p=100$.   Very roughly,  for $n << p$ we find that
ridge regression (RR) works comparatively better but that, as
sparsity increased,  EN was often best.   As $n$ increased but remained less than $p$,
EN improved relative to RR and was often the best.   When $n > p$,  LM's were usually
best with little or no perturbation of the data but  performance 
rapidly deteriorated as perturbations
or sparsity increased.  Also, a version
of SCAD often performed best.    For penalty selection in practice, however,
these conclusions must be tempered by the
more detailed analyses given in Sec.  \ref{results1} where we include considerations
from variable selection and the intuition developed in Sec. \ref{realdata} to address model
mis-specification

A second contribution of this paper is to observe that the OP is actually relatively common;
see Sec. \ref{theory1}.
So,  we needn't limit ourselves to the specific shrinkage methods that have been studied.
When $n > p$, we search a class of shrinkage methods that have the OP; when
$n < p$ we search a class of shrinkage methods that do not necessarily have
the OP; see Subsec. \ref{usingGAs}.
In either case, we generate priors (or penalties) that are optimal for a given data set,
again in a predictive instability sense. 
In practice, this means we can use a technique such as genetic algorithms (GAs) to optimize a
fitness function corresponding to good prediction to find optimal priors.  
We find that GA optimization for use in prior selection benefits from choosing
point estimates for the $\beta_j$'s when $p < n$.  These estimates
can always be assumed to exist and our theory from Sec. \ref{theory1}
carries over to this setting.  That is, we `shrink' to estimated
values in the James-Stein sense to get shrinkage estimates in the penalized likelihood
sense.

Thus, a third contribution of this paper
is  the use of GAs and part of the data to form a prior that we then
apply predictively.   As seen in Sec. \ref{GAstuff}, using both GAs and shrinkage 
(when we can) together gives better predictive performance.
By optimizing over the prior in this way we are effectively optimizing over the penalty
and hence choosing the optimal shrinkage method.   This optimum may or may not coincide
with an established shrinkage method but will still have the OP when $n$ is large enough.
Thus, we have chosen our shrinkage method to be predictively optimal for our data.
In fact, we are de facto using GA's to obtain an approxiamtion to the
posterior based on the first part of the data for use
with the second part of the data.  We verify 
in examples that
this is predictively better than simply using all the data to make predictions.

The structure of this paper is as follows.  Sec. \ref{theory1} gives general conditions
for the OP to hold for a wide range of adaptive penalties.
We present our comparisons of existing shrinkage methods in Sec. \ref{results1}
and provide general recommendations on which to use in various settings.
We compare our recommendations to the results for a benchmark data in Sec. \ref{realdata}.
In Sec.  \ref{GAstuff}, we present our GA optimization verifying that the theory in
Sec. \ref{theory1} holds and the result of the GA is optimal given the data.
We summarize our overall findings and intuition in Sec. \ref{conclusions}.

\section{Theoretical Results}
\label{theory1}

Our proofs are motivated by the techniques in \cite{Fan:Li:2001} 
and \cite{Wang:etal:2007}.  We assume an adaptive setting
and allow different penalty functions on different parameters.
This is different from \cite{Fan:Lv:2013} who did not treat different penalty functions
on different parameters or adaptivity.   Their main result concerned
the asymptotic equivalence of penalized methods and permitted
$p$ to increase with $n$.  Like other papers that allow $p$ to increase with
$n$,  some of the conditions appear artificial.   For example, aside from
truncations of  the parameter space, one must assume that there is a sequence
of explanatory variables such that if one of the early variables is correct
and accidentally not included it can be reconstructed from later explanatory variables
in the sequence, thereby sacrificing identifiability.  Our result, like many others, assumes
either fixed $p$ or $p$ increasing so slowly with $n$ that the required convergences hold.

\subsection{General Penalized Log Likelihood}
\label{genloglike}

Write the linear model as
 \begin{equation}
\label{reg} 
 Y_i = x_i ' \beta + \epsilon_i, \hspace{.25 in }   i=1, \ldots, n,
 \end{equation}
 where $x_i = (x_{i1}, \ldots , x_{ip})'$ is the p-dimensional covariate, $\beta = (\beta_{i1}, \ldots , \beta_{ip})'$ is the vector of associated regression coefficients, and $\epsilon_i$ are IID random errors with median 0. Assume that $\beta_j \neq 0$ for $j \leq p_0$ and $\beta_j =0 $ for $j > p_0$ for some $p_0 \geq 0 $. Here, we regard the $x_i's$ as deterministic. When needed we write $\beta \in \Omega \subset \mathbb{R}^p$ where $\Omega$ is open and $\Omega = \bar{\Omega}^0$. 
  
Let $ (x_i, Y_i) $ for $i=1, \ldots , n$ each have density $\rho(Y_i |x_i, \beta)$ (with respect to a fixed dominating measure)  such that the six regularity conditions stated below are satisfied.  Let $L(\beta|x^n)$ be the log-likelihood function of the observations $(x_1, Y_1) , \ldots , (x_n, Y_n) $ and denote the penalized log-likelihood objective function as
$$
Q(\beta) = L(\beta|x^n) + n\sum^n_{i=1}\lambda_j f_j(\beta_j).
$$

Here we write $x^n$ to mean $x_1, \ldots, x_n$ for ease of notation. Let \begin{align*}
&a_n = \text{max} \{ \lambda_j (\mathcal{D}_n): 1 \leq j \leq p_0 \},\\
&b_n = \text{min} \{ \lambda_j (\mathcal{D}_n): p_0 < j \leq p \}.
\end{align*}
Below we state our six regularity conditions.

\begin{description}

\item[Condition 1]\label{cond1}
Each Fisher information matrix $I(\beta | x_i)= -E\left(\frac{\partial^2 }{\partial^2 \beta} ln\; \rho(Y_i | x_i, \beta)\right)$ exists and is positive definite uniformly in $i$, i.e. is bounded above and below. Also for some $B, b > 0$,  we have $BI_{p\times p} \geq I(\beta |x_i ) \geq bI_{p\times p}$. 

\item[Condition 2]\label{cond2} The log likelihood of $\rho(Y_i | x_i, \beta)$ has a convergent second order Taylor expansion. That is, for all $j, \ell = 1, \ldots, p$, we have that $\forall \beta \in \Omega$, $\forall x_i$ and as $\eta \rightarrow 0$,
$$ 
E\left[\sup_{\beta \in \cal{B}(\beta_0,\eta)} \left| \frac{\partial^2}{\partial \beta_j \partial \beta_{\ell}}ln\; \rho(Y_i| x_i, \beta) \right| \right] \rightarrow I(\beta_0|x^n),
$$
where $\cal{B}(\beta_0,\eta)$ is the ball centered at $\beta_0$ with radius $\eta>0$
in Euclidean distance.

\item[Condition 3]\label{cond3} For any $\epsilon$, $\exists N $ such that $\forall n \geq N$ 
$$  
\sup \left| \frac{1}{n} \sum^n_{i=1} I(\beta | x_i) - I(\beta| x^\infty) \right| < \epsilon
$$ 
where $I(\beta | x^\infty)$ is positive definite.

\item[Condition 4]\label{cond4}
There exists an increasing sequence of compact sets $C=C_n$ in the parameter space and constants $M=M_n \in \mathbb{R}^+$ such that for all $n$, $\sup_{\beta_j \in C_n} |f_j'(\beta_j)| \leq M_{n}$. That is, the penalty term has a uniformly bounded first derivative in $j$. 

\item[Condition 5]\label{cond5}
For $\epsilon $ in a neighborhood around zero, $\sup_j \sup_{|\epsilon| \rightarrow 0} f'_j(\epsilon) =0$. Consequently, if $\beta_j=0$, then we have $f'_j(0)= \sup_{|\epsilon| \rightarrow 0} f'_j(\epsilon) =0$ uniformly in $j$. 

\item[Condition 6]\label{cond6}
There exists uniform Taylor expandability for 
$f_j(\beta_j)$. That is, for $h\in \mathbb{R}$, $f_j(\beta_0 +h) - f_j(\beta_0) = f_j'(\beta_0)h +o_j(1)$ uniformly in $j$ where $\sup_j o_j(1) \rightarrow 0$. 

\end{description}

Our main result for penalized likelihoods is the following.

\begin{thm}(Oracle Property)
\label{penllOP}
Assume Conditions 1--6 hold. Suppose $\sqrt{n}a_n \rightarrow 0$ in which $a_n$ 
satisfies $a_n = \frac{1}{h(n)\sqrt{n}}$ where $\frac{1}{h(n)\sqrt{n}} \rightarrow 0$ 
and $h(n) \rightarrow \infty$, and 
suppose $\sqrt{n}b_n\rightarrow \infty$ in which
$b_n = \frac{g(n)}{\sqrt{n}}$ where $\frac{g(n)}{\sqrt{n}} \rightarrow 0$ 
and $g(n) \rightarrow \infty$. Then, the estimator $\hat{\beta}=(\hat{\beta_{1}}',\hat{\beta_2}')'$
 satisfies $P(\hat{\beta_2}=0) \rightarrow 1$ and 
$\sqrt{n}\hat{I}_{1}(\beta_{1}| x^n) (\hat{\beta_{1}} - \beta_{1}) \rightarrow \mathcal{N}(0 , I_1(\beta_{1}|x^n))$, where $\beta_1$ is the component of $\beta$ containing all the nonzero elements.
\end{thm}

For proof see Supplement A, Sec. \ref{penlltheory}.

Therefore, we can construct an oracle procedure by choosing any $f_j(\beta_j)$ that satisfy Conditions 4 and 5.  We can use the same function $f_j(\beta_j)$ for each $j$ , if we want, or we can have a different penalty for each parameter. Thus, we are also not restricted, at least theoretically, in our choice of the shrinkage methods except perhaps by requiring the OP.

\subsection{Penalized Empirical Risk}
\label{penemprisk}

We now extend the OP results for penalized likelihoods to penalized empirical risks. This is arguably a more general setting because we allow a wider range for the function of the data as well as an arbitrary penalty function. 

Consider the same regression scenario as Section \ref{genloglike} and a distance
$d(y_i-x_i'\beta)$. 
Let $R(\beta | x^n)= \frac{1}{n} \sum^n_{i=1} d(y_i -x_i ' \beta)$ be 
the empirical risk of the observations $(x_1, Y_1), \ldots , (x_n, Y_n)$  where 
$x^n$ is shorthand notation for $x_1, \ldots, x_n$ and denote the penalized 
empirical risk objective function as
$$
Q(\beta) = R(\beta | x^n) + n \sum^p_{j=1} \lambda_j f_j(\beta_j).
$$
%Let 
%\begin{align*}
%&a_n = \text{max} \{ \lambda_j (\mathcal{D}_n): 1 \leq j \leq p_0 \},\\
%&b_n = \text{min} \{ \lambda_j (\mathcal{D}_n): p_0 < j \leq p \},
%\end{align*}
% and $\alpha_n = n^{-\frac{1}{2}} + a_n$. 
%% also don't know why \alpha_n is defined
We specify some similar but slightly different conditions before getting into the proofs. These conditions are similar to those in Section \ref{genloglike}, but they are different and necessary for the cases we discuss in this section. 

\begin{description}[resume]

\item[Condition 7] The empirical risk  $\frac{1}{n} \sum^n_{i=1} d(y_i-x_i'\beta )$ has a convergent second order Taylor expansion. That is, for all $j, \ell = 1, \ldots, p$ we have that $\forall \beta \in \Omega$ and as $\eta \rightarrow 0$,
$$ 
E\left[\sup_{\beta \in B(\beta_0,\eta)} \left |\frac{\partial^2}{\partial \beta_j \partial \beta_{\ell}} \frac{1}{n} \sum^n_{i=1}d(y_i-x_i'\beta )\right | \right] \rightarrow E_L\left[R'' (\beta | x^{\infty}) \right] = I^*(\beta | x^{\infty})
$$
where $$R''(\beta | x^n)= \frac{\partial ^2}{\partial^2 \beta} R(\beta | x^n) =  \frac{\partial ^2}{\partial^2 \beta}\left( \frac{1}{n} \sum^n_{i=1} d(y_i-x_i'\beta )\right) 
$$
and for some $B, b > 0$,  we have $BI_{p\times p} \geq I^*(\beta |x^{\infty} ) \geq bI_{p\times p}$. In other words, $I^*(\beta | x^{\infty}) $ is positive definite and we abbreviated it to $I^*(\beta)$. 
\end{description}

\noindent We see that $I^*(\beta )$ is a sort of analog to the Fisher information matrix in the context of empirical risks. It is not necessarily the Fisher information matrix; however, if our empirical risk happens to correspond to a log-likelihood, then it is possible that $I^*(\beta ) = I(\beta)$.
\begin{description}[resume]

\item[Condition 8]
There exists an increasing sequence of compact sets $C=C_n$ in the parameter space and constants $M=M_n \in \mathbb{R}$ such that for all $n$, $  \sup_{\beta_j \in C_n} |f_j'(\beta_j)| \leq M_{n}$. That is, the penalty term has a uniformly bounded first derivative in $j$. 

\item[Condition 9]
For $\epsilon $ in a neighborhood around zero, $ \sup_j \sup_{|\epsilon| \rightarrow 0} f_j'(\epsilon) =0$. Consequently, if $\beta_j=0$, then we have $f_j'(0)=  \sup_{|\epsilon| \rightarrow 0} f_j'(\epsilon) =0$ uniformly in $j$. 

\item[Condition 10]
There exists uniform Taylor expandability for $f_j(\beta_j)$. That is, for $h\in \mathbb{R}$, $f_j(\beta_0 +h) - f_j(\beta_0) = f_j'(\beta_0)h +o_j(1)$ uniformly in $j$ where $\sup_j o_j(1) \rightarrow 0$.

\end{description}

\noindent Again the notion of uniform Taylor expandability in $j$ here requires the error terms $o_j(1)$ in Condition 10 go to zero at the same time.

The following lemma identifies the class of distance functions we consider.

\begin{lemma}
\label{exp_dprime}
Let $d(u)$ be an even distance function with a unique minimum at 0. If $u$ comes from some distribution with pdf $f_U(u)$ that is symmetric about zero with support $[-a, a]$ for $a \in \mathbb{R}$, then $E_U[d'(u)] = 0$.  \end{lemma}

\begin{proof}
Since $d(u)$ is an even function, we know $d'(u)$ is an odd function. By definition of expectation, 
$$ E_U(d'(u)) = \int^a_{-a} d'(u) f_U(u) du.  $$

\noindent Note that since $f_U(u)$ is a symmetric distribution, $f_U(u) = f_U(-u)$, so $f_U(u)$ is an even function. Let $g(u) =  d'(u) f_U(u) $. Then $g(u)$ is an odd function, and we have 
$ \int^a_{-a} g(u) du = 0,$ so 
$$ E_U(d'(u)) = \lim_{a \rightarrow 0} \int^a_{-a} g(u) du = 0 .$$
\end{proof}

Note that Lemma \ref{exp_dprime} is true for any even distance function $d(u)$, and this is useful for us because we focus on the distance function $d(y_i - x_i'\beta)$ which is even due to the symmetry of $(y_i -x_i'\beta)$.

Our main result for penalized empirical risks is the following.

\begin{thm}(Oracle Property)
\label{empriskOP}
Assume Conditions 7--10 and Lemma \ref{sparse} hold. Suppose $\sqrt{n}a_n \rightarrow 0$ where $a_n$ satisfies $a_n = \frac{1}{h(n)\sqrt{n}}$ such that $\frac{1}{h(n)\sqrt{n}} \rightarrow 0$ and $h(n) \rightarrow \infty$, and suppose $\sqrt{n}b_n\rightarrow \infty$ where $b_n$ satisfies $b_n = \frac{g(n)}{\sqrt{n}}$ such that $\frac{g(n)}{\sqrt{n}} \rightarrow 0$ and $g(n) \rightarrow \infty$. Then the estimator $\hat{\beta}=(\hat{\beta_{1}}',\hat{\beta_2}')'$ satisfies $P(\hat{\beta_2}=0) \rightarrow 1$ and $\sqrt{n}\hat{I^*}_{1}(\tilde{\beta}_{10}| x^n) (\hat{\beta_{1}} - \beta_{10}) \rightarrow \mathcal{N}(0 , I_1(\beta_{10}|x^n))$ .
\end{thm}

For proof see Supplement A, Sec.  \ref{penrisktheory}.

Taken together, Theorems \ref{penllOP} and \ref{empriskOP} give us insight into how large is the class of oracle procedures. Previously it seemed oracle procedures were rare, isolated choices of priors. Now, even though we have not characterized the class of all oracle procedures, we can see that the conditions for a procedure to have the oracle property are quite general, allowing a large range of likelihoods, distances, and priors. 

Since there are obviously many oracle procedures, the key question becomes which one to choose in a given context. We propose choosing a 'best' oracle procedure by optimizing a stability criterion over both a class of penalty functions and a class of distance functions. Allowing for $f_j(\beta_j)$ to be a different penalty for each parameter, as long as they are uniformly Taylor expandable and have similar properties, is powerful because we can choose variable-dependent penalties. This is a more general sense of adaptivity than each $\beta_j$ merely having its own shrinkage parameter $\lambda_j$. We believe a natural choice for $f_j(\beta_j)$ might be something like $|\beta_j|^{1-|corr(y_i,x_i)|}$ because it is data driven. We have not done this here, and we leave it for future work.  Here we focus on comparing existing methods and then using genetic algorithms to find optimal choices for the $f_j(\beta_j)$'s.

\section{Computational Comparisons}
\label{results1}

Every shrinkage method for linear models generates a predictor of the form
\eqref{predictorlm} ,i.e., $\hat{Y}(x_{n+1}) = x_{n+1}^T \hat{\beta}$ where the
estimate $\hat{\beta}$ of $\beta$ comes from choosing $\lambda$ and the 
$w_j$'s (or other data-driven parameters).
It is well-known that many shrinkage methods (LASSO, EN, etc)
zero-out coefficients $\beta_j$ 
and so do variable selection as well as estimation.  Here, we
look only at the instability of predictive error and the accuracy of variable selection.

Following \cite{Luo:etal:2006} we add random
$N(0, \tau^2)$ noise to the $y_i$'s in ${\cal{D}}$ and
denote the partition of the perturbed data by 
${\cal{D}}_n(\tau) = {\cal{D}}_{train}(\tau) \cup {\cal{D}}_{test}(\tau)$.
For any predictor $\hat{Y}$, we define its instability to be
$$
S(\hat{Y})_\tau = \sqrt{\frac{1}{n_{test}} \sum_{i \in {\cal{D}}_{test}(\tau)} (y_i - \hat{Y}_\tau (x_i) )^2 }
$$
where $\hat{Y}_\tau$ means we have formed a predictor using ${\cal{D}}_{train}(\tau)$.
In the computations we present in this section we used $\tau_k = .2 k$, $k=1, \ldots, 8$,
to generate instability curves of the form $(k, S(\hat{Y})_{\tau_k})$,  and looked for patterns.

Intuitively, perturbing the $Y$'s by adding normal noise should only
increase $S(\hat{Y})_{\tau_k}$, i.e., in curves should generally increase with $\tau$.
Of course, we prefer instability curves that are small -- less instability upon
perturbation suggests a better predictor.  However,
if a instability curve decreases with $\tau$ then perturbation of $Y$
is making the predictor more stable.   We take this to mean the predictor
is discredited for some reason.   We suggest this behavior arises when
the predictor has omitted or included terms incorrectly
or has poorly chosen coefficients.
We prefer predictors with instability curves that are 
lower than the instability curves of competing predictors and smoothly increase slowly
with $\tau$.  

Our basic computational procedure is as follows.   Fix a number $K$ of values of $\tau$
to form the points on the instability curve and a (large) number $L$ for the number
of iterations to be averaged.   For each $k=1, \ldots , K$, let 
 $\ell = 1, \ldots, L$.  Instability curves for a given predictor can generically be
formed by the following steps.

\begin{enumerate}

\item  For each $\ell$, randomly split ${\cal{D}}_n$ in to ${\cal{D}}_{train}$ and ${\cal{D}}_{test}$.

\item Perturb the $y$-values in ${\cal{D}}_{train}$ and ${\cal{D}}_{test}$
using $N(0, \tau^2_k )$ noise. Call the results ${\cal{D}}_{train,  \tau_k}$ and
${\cal{D}}_{test,  \tau_k}$, respectively.

\item  Using ${\cal{D}}_{train,  \tau_k}$ form the competing predictors denoted $\hat{Y}_{\tau_k}$.

\item Using ${\cal{D}}_{test,  \tau_k}$ obtain $S(\hat{Y}_{\tau_k})$ for each predictor.

\item Let
$S(\hat{Y}_{\tau_k , \ell})$ be the $\ell$-th value of $S(\hat{Y}_{\tau_k})$.

\item For each predictor and $k$, find the sample mean:
$$
S_{\tau_k}(\hat{Y}) = \frac{1}{L} \sum_{\ell=1}^L S(\hat{Y}_{\tau_k, \ell}).
$$

\item Plot $S_{\tau_k}(\hat{Y})$ as a function of $\tau_k$ for each predictor.

\end{enumerate}

In this section all simulated data come from 
$$
Y = X\beta + \epsilon
$$
where the rows $X_i=x_i$ of $X$, for $i=1, \ldots , n$, are $MVN_p(0, M)$ for various choices of variance matrix 
$M$  or are IID $\sim t_3$ to see the effect of heavier tails.  The parameter values for
$\beta \in \mathbb{R}^p$ are IID $N_p(4,1)$.  So,
$\dim(Y) = n$, $\dim(X) = n\times p$,  and $\dim(\epsilon) = n$.
We use two choices for the distribution of $\epsilon_i$,  $N(0, 1)$ and $t_3$, to represent light and
heavy tails in the error, respectively.  
We are concerned mainly with the case $p<n$, but include cases $p>n$ for completeness.

We compared predictors from ten different shrinkage methods as well as a full linear model.  
Seven of the shrinkage methods
had the OP, namely,  LAD-LASSO (uses $L^1$, in place of $L^2$ in ALASSO, Wang et al. 2007), 
ALASSO
(Zou 2006), SCAD1 ($L^1$ with a truncated $L^1$ penalty, Fan and Li 2001), 
AEN (Zou and Zhang 2009),  ASCAD1 
(Dustin 2020: $\|Y- X\beta\|_1  + \sum p_{\lambda_j}(|\beta_j|)$) i.e., an adaptive SCAD penalty
on each $\beta_j$), SCAD2 (replace $L^1$ with $L^2$ in SCAD1) ], and MCP (Zhang 2010).
Three of them did not have the OP, namely ridge regression (RR),  LASSO and EN.  

Finding the instability curves for the ten shrinkage methods
required three different pieces of software.  First,  for LASSO, RR, EN, ALASSO, AEN, 
we used the {\sf glmnet} package (see \cite{Friedman:etal:2010}) in RStudio  Ver.  1.2.5033. 
Second,  for LAD-LASSO, SCAD1,  and ASCAD1 we used the
{\sf rqPen} package in R, see \cite{Sherwood:etal:2020}.   
The specific function used ({\sf cv.rq.pen()} ) 
requires $X$ to be nonsingular so we only implement these methods 
(as well as LM) in the $n>p$ case.
Third, for MCP and SCAD2 we used code that implemented a local linear approximation (LLA).  
\footnotemark

\footnotetext{Development of the LLA code was by L. Xue and reported in \cite{Fan:etal:2014}.     
We are grateful to these authors for giving us their code.}

Because we use different packages,  we had to
split the data differently for different methods.    
The LLA code requires an explicit separation of ${\cal{D}}_{train}$ into two
sets say ${\cal{D}}_{train} = {\cal{D}}_{train, \lambda} \cup {\cal{D}}_{train, \beta}$ so that
a small portion of the data is set aside for estimating $\lambda$.   However,
both {\sf glmnet} and {\sf rqPen} approximate $\lambda$ internally and we use the default 
methods of the procedures, so the splitting of ${\cal{D}}_{train}$ is done
internally to the programs.  Below we list the choices we made for the explicit splitting.

As noted in Sec. \ref{intro}, we followed \cite{Zou:2006} for the adaptive methods
and chose $\hat{w}_j = 1/ |\hat{\beta}_{j, OLS}|$ for $p < n$.
When we do not have enough data to implement OLS, i.e. $p > n$, we used $1/|\hat{\beta}_{j, SCAD2}|$.  In fact, the computations can be done, at least in principle, using any $\sqrt{n}$-consistent estimator 
of the $\beta_j$'s for the $w_j$'s. The reasoning behind this choice comes 
from the assumptions on $a_n$ and $b_n$ in Theorems \ref{penllOP} and \ref{empriskOP}.
%We comment that using $L^2$ in LAD-ALASSO, is just ALASSO and
%if we tried to do ASCAD2 we couldn't assign $w_j$'s because the
%LLA code didn't allow for that.   

We focus here on a high dimensional setting of $p$ relative to $n$. 
Thus, we set $p=100$ and consider three sparsity levels,10\%, 50\%, and 90\%,
and four sample sizes $n=40, 75, 150, 500$. For each $n$ we let $L=100$ for the instability computations. That is we average over 100 datasets to get a instability value for each perturbation level. 
We use a training data set to form the predictor and a testing data set to evaluate its performance, and for the various sample sizes  we set 
$({\cal{D}}_{train},  {\cal{D}}_{test})$ = $(28 , 12)$, $(60, 15)$,  $(120, 30)$,  and $(350, 150)$,
respectively.  For the $LLA$ methods we allocate 10\% of the ${\cal{D}}_{train}$ data to ${\cal{D}}_{train, \lambda}$ and the remaining $90\%$ to  ${\cal{D}}_{train, \beta}$.   Our choices
respect the fact that for the methods to be comparable, they must be trained with exactly 
the same data, and evaluated on a hold out set that has not been used in the training process. 

The next four subsections present our simulation results for the four sample sizes
and three sparsity levels.
The final subsection summarizes our recommendations before we turn to a
benchmark data example in Sec. \ref{realdata}.

\subsection{Sample Size $n=40$}
\label{n=40}

Our first example uses the value $n=40$ that is relatively small compared to $p=100$.
Here, we examine how good the variable
selection, estimates, and thus predictions are for various sparsity levels
where sparsity means the percentage of zero coefficients in the model
used to generate the data. 

We use two assessments for this.  First, we generate instability curves.  Then we also
look `inside' the predictor to see which variables were included correctly or
incorrectly.  These two assessments are roughly analogous to variance and bias.  
However,  as noted in Sec. \ref{intro},we think
of instability curves as more important because they reflect `variability' and
partially reflect bias:
If a predictor includes a variable
incorrectly it can still be downweighted by small parameter estimates and if the
predictor fails to include a correct variable (that is important) the instability curve
can increase to detect it.

Fig. \ref{Fig_n=40} shows that there is overall more instability with heavy tails, and, independent of heaviness 
of the tails, as sparsity increases the methods become more stable.  We observe two clusters of methods. Namely, the non-adaptive methods in the lower cluster (L,EN,RR) and the adaptive methods in the upper cluster(AL,AEN,SCAD2, MCP). This is consistent with the fact that we have far fewer observations than predictors, so estimating the extra hyperparameters accurately is difficult.  Overall, this figure suggests 
that RR and EN perform best for sparsity levels 0.1 and 0.5 while EN is best for sparsity 0.9, although RR is not discredited, at least for light tails.

 This is seen for both light and heavy tails. A possible explanation for the
good performance of RR is that there is not enough information in the small sample 
 relative to the number of predictors to set any coefficients to zero. Thus RR performs well because it does not exclude any variables. EN, a mixture of LASSO and RR,
also performs well because it may set fewer coefficients to zero than other non-adaptive methods. 

It may seem odd to say RR is not discredited in the upper right panel of
Fig. \ref{Fig_n=40}.  However, all the other curves decrease markedly as the perturbation
increases suggesting they have chosen poor models.   Also, RR is smallest at perturbation zero.

The curves increase after the
perturbation is large enough that it outweighs the poor selection of the models.
Also, it is seen that there are some jumps in the curves.  These are usually at the beginning
when the perturbation changes from zero to a positive number.  We suggest these jumps
indicate a lack of stability under perturbation meaning that the shrinkage method they
represent is unstable in terms of choosing a good predictor.

\begin{figure}[htp!]
\begin{center}
\begin{tabular}{ccc}
%\subfloat[scatter plot of YIELD vs. TKWT]{
\includegraphics[width=.35\columnwidth]{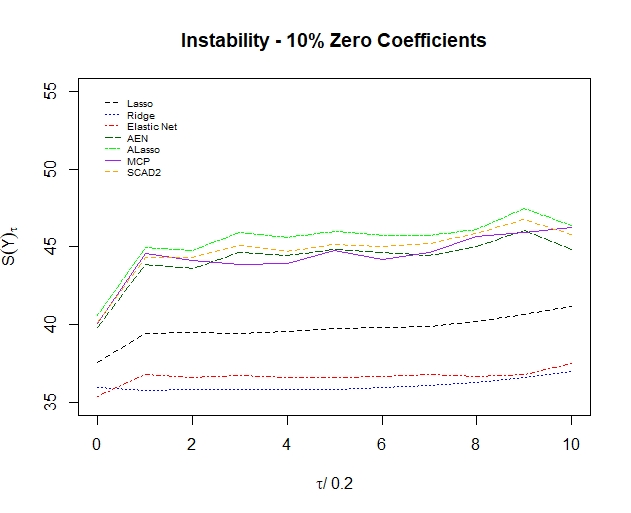}
%\hfill
%\subfloat[scatter plot of YIELD vs. SPSM]{
\includegraphics[width=.35\columnwidth]{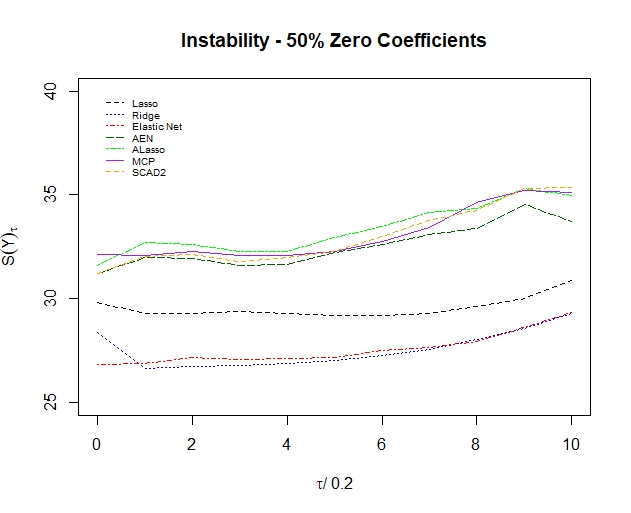}%}%
%\hfill
%\subfloat[scatter plot of YIELD vs. KPS]{
\includegraphics[width=.35\columnwidth]{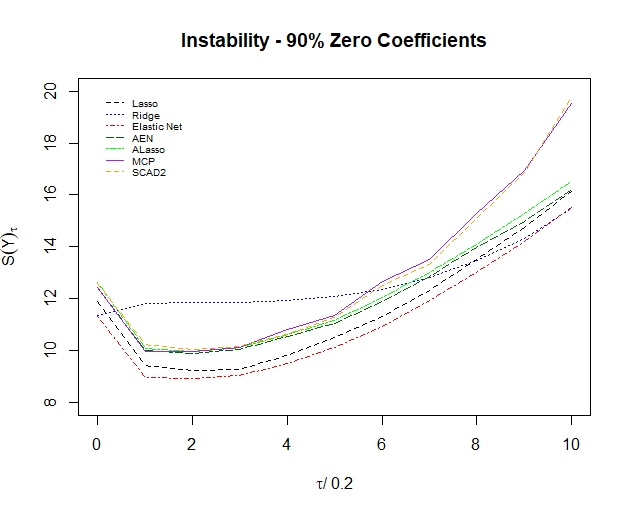}%} 
\\
%\subfloat[scatter plot of YIELD vs. TKWT]{
\includegraphics[width=.35\columnwidth]{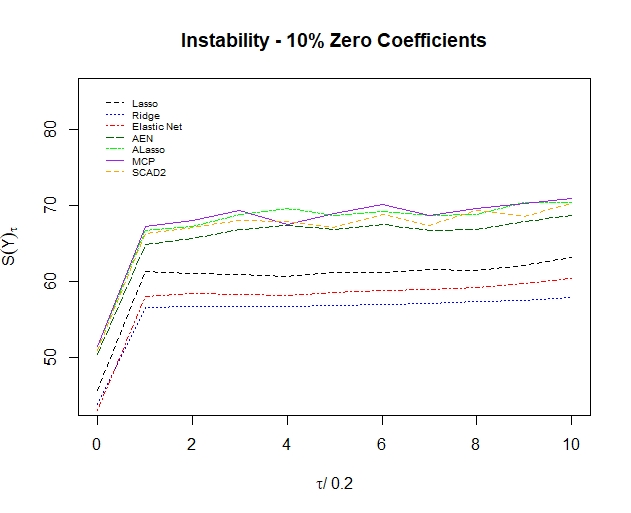}
%\hfill
%\subfloat[scatter plot of YIELD vs. SPSM]{
\includegraphics[width=.35\columnwidth]{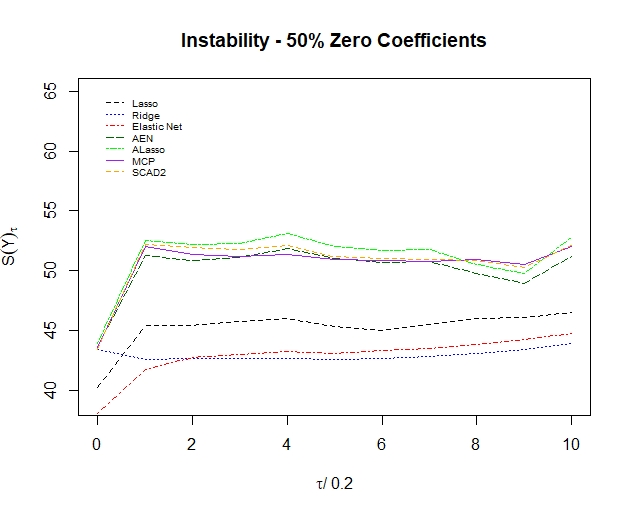}%}%
%\hfill
%\subfloat[scatter plot of YIELD vs. KPS]{
\includegraphics[width=.35\columnwidth]{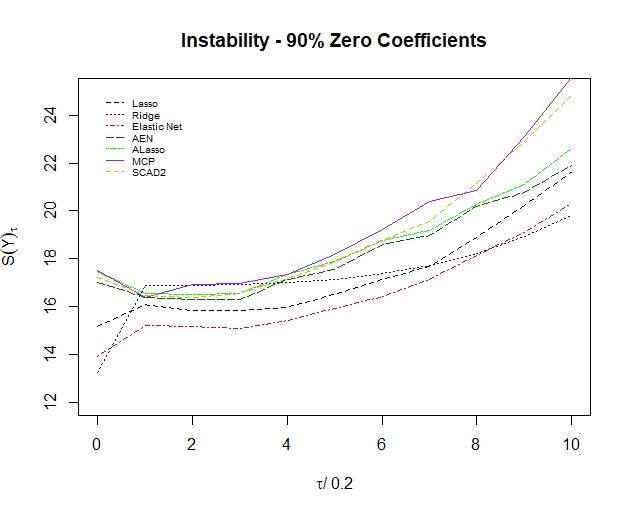}%}
\end{tabular}
\end{center}
\caption{$n=40$ From left to right the sparsity increases from 10\%, to 50\%, to 90\%.  
From top to bottom
the heaviness of the tails of $X$ and $\epsilon$ increases from normal to $t_3$.  The distribution
of the 100 IID outcomes for the parameter $\beta$ is $N(4, 1)$.  In this simple case,
we used the identity matrix in the normal to generate the matrix $X$.}
\label{Fig_n=40}
\end{figure}

Next,  we look at the 
predictors themselves, on average, to help us interpret the instability curves. 
Table \ref{Table_n=40} compares 
each of the methods for both light (Li) and heavy (H) tails based on the total percentage of coefficients 
the methods sets to 0 on average (Tot; ideally equal to the sparsity level),
what percentage of those set to 0 were truly 0 (Tr; ideally 100\%) and what 
percentage of those set to 0 were truly non-zero (Fa; ideally zero). 
Note that we are focused on the 
sparsity in this setting, so we frame these tables in terms of the percentage of variables excluded by a method. Instead, if we wish to describe the variables included, we simply subtract each entry in the table from 1.

 \begin{table}
 \centering
\begin{tabular}{llllllllllll}
\hline
  \multirow{1}{*}  Sparsity & 
\multicolumn{3}{c}{.1}&
\multicolumn{3}{c}{.5}&
\multicolumn{3}{c}{.9}& \\
      &       & Tot & Tr & Fa & Tot & Tr & Fa & Tot & Tr & Fa \\
    \hline
    Li & L  &  .95   & .95  &  .95    & .92 &  .93 & .90 &  .95  &  .97 &  .76 \\
    \hline
    Li & AL &   .88  & .85 & .88 &.89  &  .90   &  .87  &  .91 & .95    &  .56  \\
    \hline
    Li & EN &  .87  & .86     & .87 &  .82 &   .84 & .80 &  .91  & .94   &  .56  \\
    \hline
Li &AEN &   .88 & .84   & .88 & .88  & .90&  .80  &  .90  &  .94   & .55   \\
\hline 
Li &SCAD2 & .82&  .78 & .82  & .82 & .84   &  .80  &  .87  &  .91 &  .49  \\
\hline 
Li & MCP & .83 &  .78  & .83& .84 & .85 &  .82  & .88    &  .92  &  .50   \\
\hline 
\hline
H & L & .97   &  .98    &   .97    &  .93    & .95  &  .90      & .96  & .97  &  .88\\
\hline 
H & AL &.89   &  .92     &   .89    & .90  &   .93 &  .88       & .91 & .99   &  .69  \\
\hline 
H & EN &   .93   &  .96     &   .93      &  .88  & .91     &  .85  &   .88  & .93    &  .72     \\
\hline
H & AEN &   .87 &    .90      &  .88       &  .90 & .93     &  .87        &     .90  & 1    &  .66      \\
\hline
H & SCAD2& .83   &     .85       &    .82  & .84&  .86 &  .81  &  .87 & .86   &  .59  \\
\hline
H &MCP &    .84  &     .86   &   .84      &  .84  &  .88  & .82    &  .88& .90   &  .63    \\
\hline 
\end{tabular}
\caption{ Variable selection performance for $n=40$}
\label{Table_n=40}
\end{table}

From Table \ref{Table_n=40}, it is seen that none of the methods performs well.    
The values in the Tot columns for low and medium sparsity are much too high which leads to high values in the Tr columns, trivially.  The high values in the Fa columns indicate that all methods are
typically excluding far too many terms.  
There are some very mild exceptions -- SCAD2 or MCP -- that perform noticeably
better than other methods but even so their performance is poor.  For high sparsity
the methods do exclude fewer correct variables than for low sparsity.
Table \ref{Table_n=40} agrees with Fig. \ref{Fig_n=40} in that variable selection is generally 
worse for heavier tails.   Curiously, the better performing methods in terms of
variable selection are different from the better performing methods in terms of
instability curves.  We attribute the difference to how well coefficients are estimated
rather than to how well variables are selected.  Indeed,  Table \ref{Table_n=40}
shows that all methods usually selected far fewer variables than necessary,
e.g., for L with .1 sparsity, five variables were included on average
but 90 should have been.

We also considered two cases where the $x_j$'s have a nontrivial dependence structure.
Specifically, we set $X \sim MVN_{100}(0, M)$ where $M$ is tridiagonal or Toeplitz,
both with light tailed errors.    We do not show the plots for these cases, howevver we give
our recommendations in Subsec.  \ref{Summarytables}.

As in the independence case, the methods performed generally better as sparsity increased.
For figures analogous to those in Fig. \ref{Fig_n=40} but for tridiagonal variance
matrices, we found that for low and medium sparsity,  EN and RR were
typically best and close to each other.   For high sparsity we got the same results
as in the upper right panel of Fig. \ref{Fig_n=40}.   When the variance matrix of
$X$ was Toeplitz, we found that EN is performing best for all low and medium sparsity levels.
For high sparsity SCAD2 and MCP performed nearly the same as EN.  Interestingly,
SCAD2 and MCP were amongst the worst performing methods in the
independence case but with the Toeplitz structure they were only slightly worse than EN.
We suggest that the degree of dependence in the tridiagonal case is not much
more than in the independence case.  
We also suggest that when there is
high enough dependence e.g., the Toeplitz case,  penalties that have some data dependence
(but not too much) such as EN, SCAD2, and MCP are best able
to generate predictions.

In terms of variable selection as indicated in Table \ref{Table_n=40}, 
we found that the results for the tridiagonal and Toeplitz were similar
to the independence case across all methods and sparsity levels
except for EN in the Toeplitz case where EN performed noticeably better than the other methods across all sparsity levels.  In addition,
L does well in the high sparsity case.   We suggest
that being able to choose what the penalty looks like as in EN gives
some advantage over the other methods.  We return to this point in
Sec. \ref{GAstuff}.

\subsection{Sample Size $n=75$}
\label{n=75}
 
Next we consider a second example where $p>n$. That is, we still have fewer observations 
than explanatory variables, but $n$ is much closer to $p$ than in Subsec.  \ref{n=40}.
Examining the instability plots in Fig. \ref{Fig_n=75}, we see that as in Fig. \ref{Fig_n=40}, 
the methods become more stable as sparsity increases and
less stable as the tails become heavier.

Similar to the $n=40$ case, we have the same clusters of methods, namely
EN, RR, and L as the lower (better) cluster with AEN, AL, MCP and SCAD2 
as the worse (upper) cluster.   EN outperforms all other methods with
low and medium sparsity regardless of the heaviness of the tails.  For high sparsity all
methods are roughly comparable except for RR which performs poorly.
That is the $n=40$ and $n=75$ cases are qualitatively similar apart from the light tail
high sparsity case with $n=40$ where model instability seems to dominate.
Specifically, for $n=40$ EN was the `best' although there was good reason to prefer RR, 
but for $n=75$ all methods but RR are essentially indistinguishable and RR performs poorly.

It is seen that the error decreases as both the sparsity and $n$
increase.  Accordingly, as both increase, the predictive error may
converge to $\sigma$, the SD of $\epsilon$.

\begin{figure}[htp!]
\begin{center}
\begin{tabular}{ccc}
%\subfloat[scatter plot of YIELD vs. TKWT]{
\includegraphics[width=.35\columnwidth]{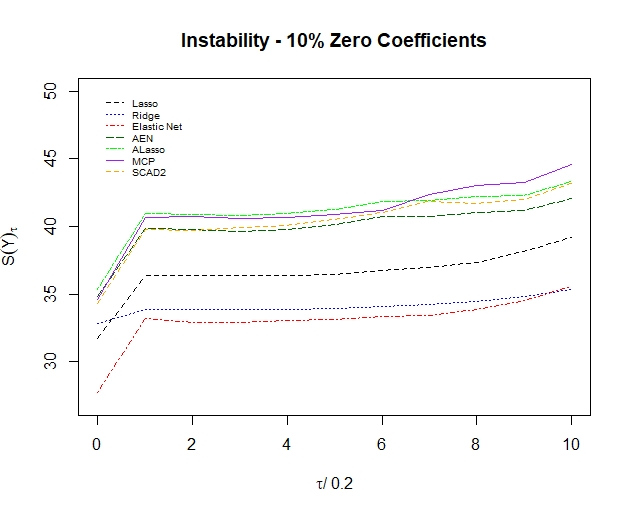}
%\hfill
%\subfloat[scatter plot of YIELD vs. SPSM]{
\includegraphics[width=.35\columnwidth]{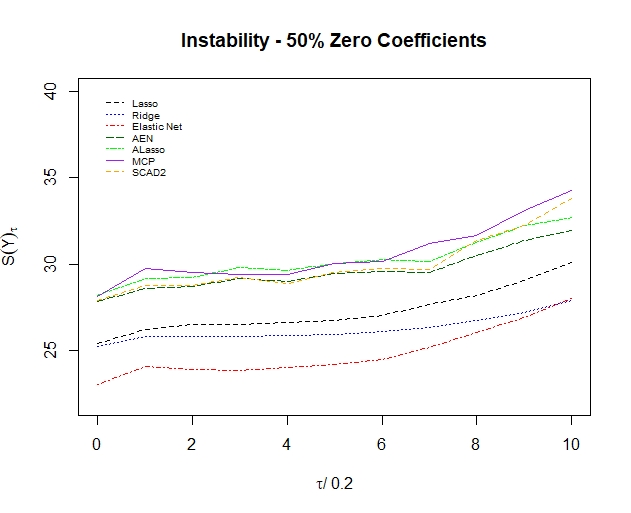}
%\hfill
%\subfloat[scatter plot of YIELD vs. KPS]{
\includegraphics[width=.35\columnwidth]{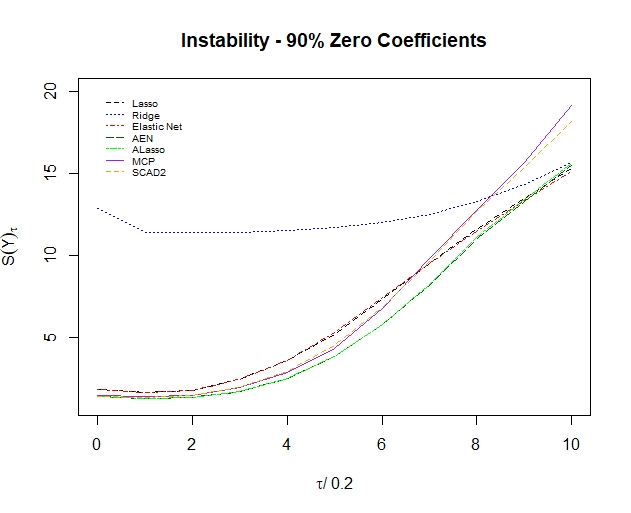} 
\\
%\subfloat[scatter plot of YIELD vs. TKWT]{
\includegraphics[width=.35\columnwidth]{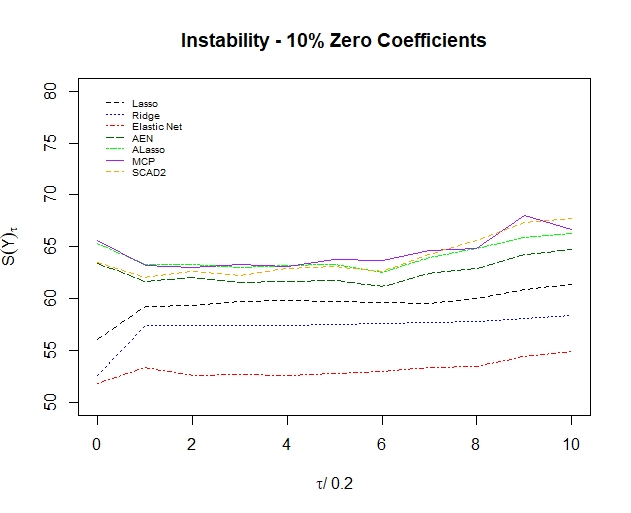}
%\hfill
%\subfloat[scatter plot of YIELD vs. SPSM]{
\includegraphics[width=.35\columnwidth]{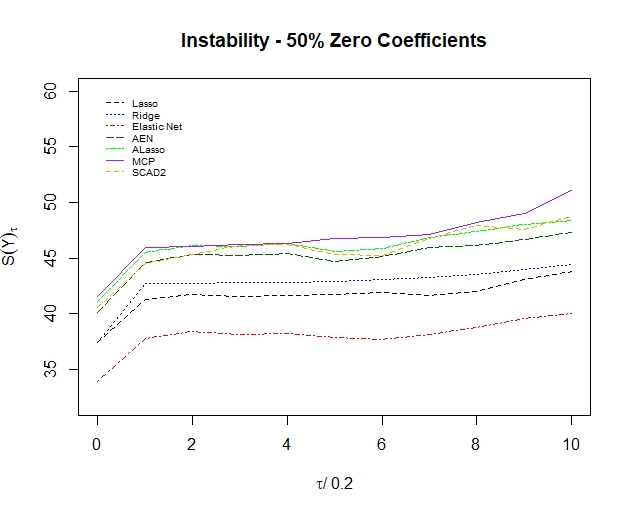}
%\hfill
%\subfloat[scatter plot of YIELD vs. KPS]{
\includegraphics[width=.35\columnwidth]{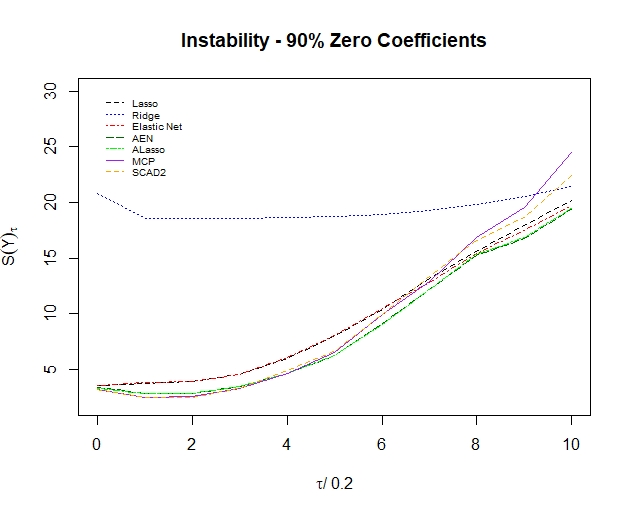}
\end{tabular}
\end{center}
\caption{$n=75$: From left to right the sparsity increases from .1, to .5, to .9.  From top to bottom
the heaviness of the tails of $X$ and $\epsilon$ increases from normal to $t_3$.  The distribution
of the 100 IID outcomes for the parameter $\beta$ is $N(4,1)$.  Again, we used
the identity matrix in the normal distribution for $X$.}
\label{Fig_n=75}
\end{figure}

 \begin{table}
 \centering
\begin{tabular}{llllllllllll}
\hline
  \multirow{1}{*}  Sparsity & 
\multicolumn{3}{c}{.1}&
\multicolumn{3}{c}{.5}&
\multicolumn{3}{c}{.9}& \\
      &       & Tot & Tr & Fa & Tot & Tr & Fa & Tot & Tr & Fa \\
    \hline
    Li & L  &  .75 & .89 &  .73   & .83  &  .92 & .75  & .76   &.84    &  0  \\
    \hline
    Li & AL & .77 &  .88 & .76   &  .79 &  .88 & .71  & .90   &  1  & 0 \\
    \hline
    Li & EN & .60 &  .78& .58  &  .70 &  .81 &   .58  &  .76 &  .85  & 0\\
    \hline
Li &AEN &   .76  &.87  & .74   &  .79 & .88 & .70  &  .90  & 1   & 0 \\
\hline 
Li &SCAD2 & .67 & .78 & .66  &  .71 & .81  & .61  &.90    &  1  & 0 \\
\hline 
Li & MCP &  .69 & .80 & .68    &  .72&  .82 &  .63  &  .89  &  .99  & 0 \\
\hline 
\hline
H & L & .90 & .92&  .90&  .85& .91  & .90  & .74  &   .83 &  0   \\
\hline 
H & AL &.83 &.99 & .82 & .83 & .90  &  .76 &.90   &  1  &   0  \\
\hline 
H & EN & .85  & .81&.85  & .77& .84  & .71  &  .74 &  .82  & 0     \\
\hline
H & AEN &  .82   &1 & .81 &.82  &.89  &  .74 & .90  &  1  &   0   \\
\hline
H & SCAD2&  .72 & .79 &.71  &.74  &.82   & .66  &.89   & .99   &  .01  \\
\hline
H &MCP &   .74 &.80 & .73 &  .75 &.82  & .68  &  .89 &  .98  &  .01   \\
\hline 
\end{tabular}
\caption{ Variable selection performance for $n=75$}
\label{Table_n=75}
\end{table}

Table \ref{Table_n=75} shows that, as with Table \ref{Table_n=40}, 
no method performs variable selection well for low to medium sparsity.   
For high sparsity, all methods were very much improved.
However, L and EN were the worst -- the only methods not having the OP.
The other methods, AL, AEN, SCAD2, and MCP,  have the OP and
perform roughly equally well.  As a generality, the methods performed
better for light tailed than for heavy tailed distributions.
As before, the better performing methods in the instability curves (EN, RR, L)
tended to perform worse in terms of variable selection and we attribute this
to better parameter estimation in the instability curves since the methods generally
included more variables than necessary.  That is, a variable may be included
with an estimated coefficient near zero.

When we included dependence via tridiagonal matrices, we found results similar to the
independent case.     Namely, the instability curves show that
EN performed best at all sparsity levels, roughly tying with most
other methods for high sparsity.  The key difference was that RR tended to perform
poorly overall.  When we switched to Toeplitz matrices, across all sparsity
levels SCAD2 was best with MCP being a close second and virtually tying
with SCAD2 with high sparsity.  EN, which was overall best for $n=40$,
was noticeably worse than both SCAD2 and MCP across all sparsity levels.
We suggest that since we have more data, the optimality of the methods
that have optimality properties (other than the OP) is effective.

In terms of variable selection, the tridiagonal case was similar to the independence
case but slightly better.   This was surprising and difficult to interpret.
In the Toeplitz case, EN, SCAD2 and MCP are noticeably better than the
other methods for low and medium sparsity.  For high sparsity, L also performs well.
This is the same as the Toeplitz case with $n=40$.  Thus, overall, the results for
dependence cases with $n=75$ are very close to the corresponding results for $n=40$.

\subsection{Sample Size $n=150$}
\label{n=150}

This is our first case where $n > p$,  making it qualitatively different from the earlier 
two subsections.    Here  we implement LM's as well as shrinkage techniques. 
Note that the ``penalty'' associated with LM's is identically zero and 
corresponds to a uniform prior. 

Fig. \ref{Fig_n=150} parallels Figs.  \ref{Fig_n=75} and \ref{Fig_n=40},
but introduces four extra methods, LM's, SCAD1, ASCAD1, and L-L.
For low sparsity, LM's start out giving the best results in predictive instability.  
However, the instability curve for LM's rises faster than for its competitors
and arguably SCAD1 is equivalent.
For medium sparsity, LM's is roughly in the middle and SCAD1 is the best.
For high sparsity, SCAD1 ties with the other top methods.  LM's
perform worse as sparsity increases.     As in earlier cases, all methods
that have nontrivial penalties
perform better for light tails than heavy tails and as sparsity increases.
Also, as the sparsity level increases, the variability amongst the methods 
decreases.  Note that RR tends to perform poorly since there is enough
data that the sparsity has an effect.

A key difference between the $n=150$ case and the $n=75, 40$ cases
is that the adaptive methods are generally performing better than the
nonadaptive methods. For instance, AEN and AL are performing better
for small perturbations than EN or L, respectively.   When the perturbations
are too high, it makes sense that the non-adaptive version of a penalty will perform
better that their adaptive versions
because they are less affected by the noise; they use fewer estimators.
One can argue that the perturbation level at which the curves for
non-adaptive penalties and their adaptive versions cross represents the
largest reasonable perturbation that should be considered for that penalty.
Moreover, the OP is not a determining factor for performance:  Some methods
with the OP perform well and some do not.   Some methods that do not have the OP
perform better than other methods that do.

\begin{figure}[htp!]
\begin{center}
\begin{tabular}{ccc}
%\subfloat[scatter plot of YIELD vs. TKWT]{
\includegraphics[width=.35\columnwidth]{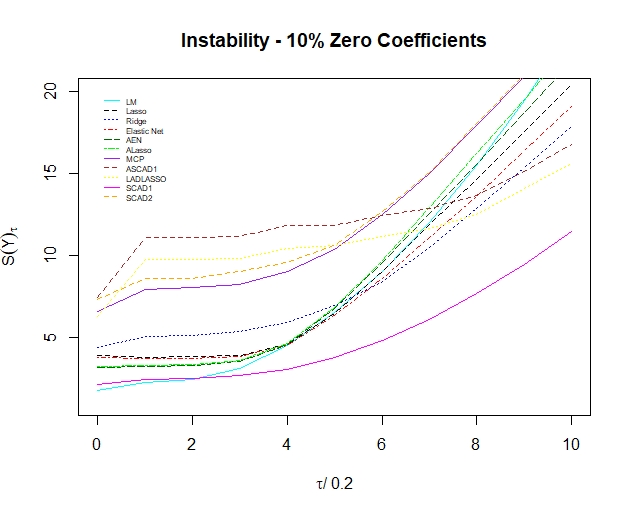}
%\hfill
%\subfloat[scatter plot of YIELD vs. SPSM]{
\includegraphics[width=.35\columnwidth]{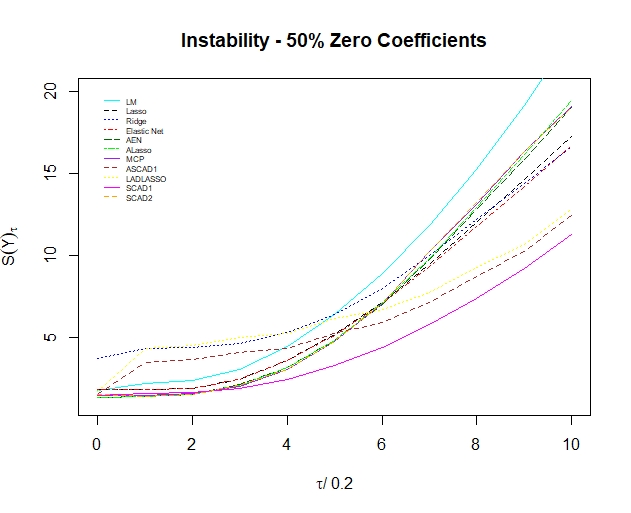}
%\hfill
%\subfloat[scatter plot of YIELD vs. KPS]{
\includegraphics[width=.35\columnwidth]{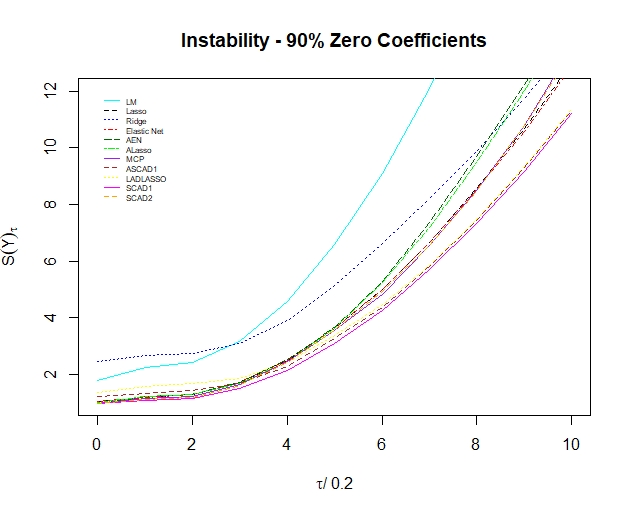} 
\\
%\subfloat[scatter plot of YIELD vs. TKWT]{
\includegraphics[width=.35\columnwidth]{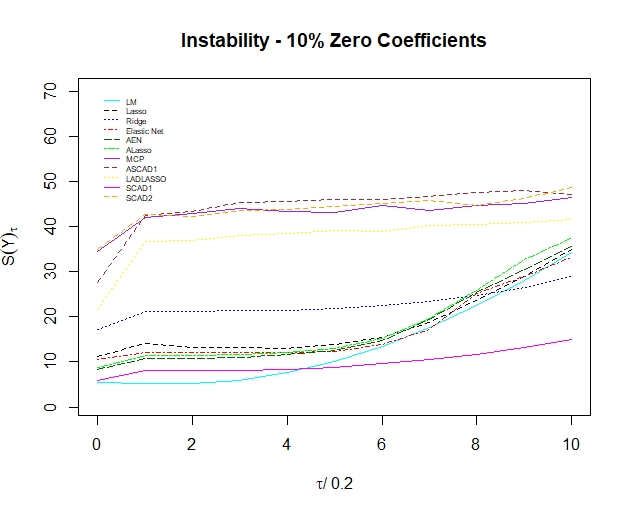}
%\hfill
%\subfloat[scatter plot of YIELD vs. SPSM]{
\includegraphics[width=.35\columnwidth]{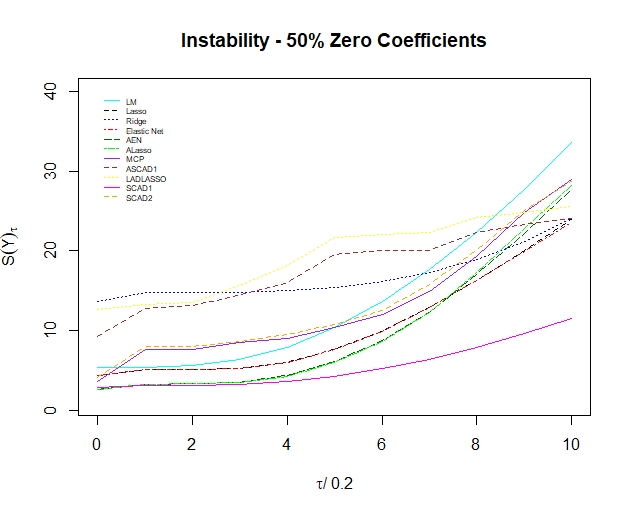}
%\hfill
%\subfloat[scatter plot of YIELD vs. KPS]{
\includegraphics[width=.35\columnwidth]{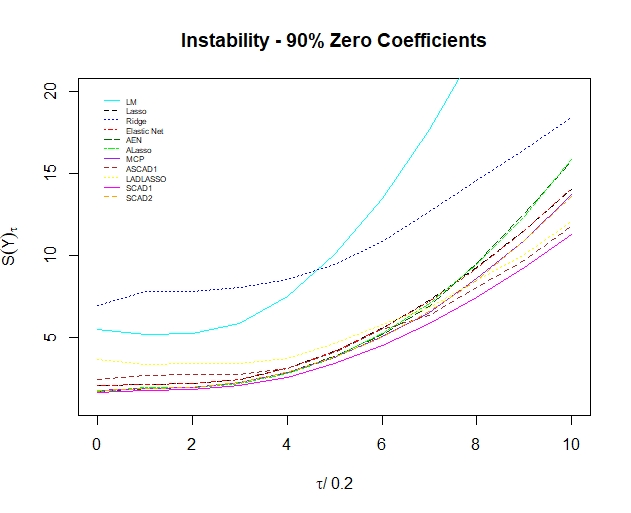}
\end{tabular}
\end{center}
\caption{$n=150$: From left to right the sparsity increases from .1, to .5, to .9.  From top to bottom
the heaviness of the tails of $X$ and $\epsilon$ increases from normal to $t_3$.  The distribution
of the 100 IID outcomes for the parameter $\beta$ is $N(4, 1)$.  We used the identity matrix
in the normal distribution for $X$.}
\label{Fig_n=150}
\end{figure}

 \begin{table}
 \centering
\begin{tabular}{llllllllllll}
\hline
  \multirow{1}{*}  Sparsity & 
\multicolumn{3}{c}{.1}&
\multicolumn{3}{c}{.5}&
\multicolumn{3}{c}{.9}& \\
      &       & Tot & Tr & Fa & Tot & Tr & Fa & Tot & Tr & Fa \\
    \hline
    Li & L  &  .04 & .38 & 0 & .30 &.60 & 0 & .81&.90&	0 \\
    \hline
    Li & AL & .11	&1&	.01 &.50&1&0 &.89&.99&0 \\
    \hline
    Li & EN &.03 &.30&0&.29 &.59 & 0&.81&.90&0 \\
    \hline
Li &AEN &  .11&1&.01 &.50&1 & 0&.88&.98&0 \\
\hline
Li & ASCAD1 & .11&.82&.03 & .48 &.95 &0 &.89&.99&0\\
\hline
Li & L-L &.01& 0 & 0 & .04 & .07 & 0 & .83 & .92 & 0 \\
\hline 
Li & SCAD1 & .03 & .30 &0 & .42 & .85 & 0 & .89 & .99 & 0\\
\hline
Li &SCAD2 & .12 & .74 &.07 & .51 & 1 & .02 & .90 & 1 & 0 \\
\hline 
Li & MCP &  .11& .67 & .05 & .50 & .99 & .01 &.89&.99&0 \\
\hline 
\hline
H & L & .04 &.33&0 & .24&.49&0 &.83&.92&0\\
\hline 
H & AL & .10&1&0 &.50&.99&0 &.90&1&0 \\
\hline 
H & EN &  .03& .28 &.43&  .24&.49&0& .82&.92&0 \\
\hline
H & AEN &  .10&1&.50 &  .50&.99&0&.90&1&0 \\
\hline
H & ASCAD1 &.13& .87 & .05 &.49&.97&.01&.90&1&.01 \\
\hline
H & L-L &.01 & .11 & 0 &.10 & .19&.01 & .87&.96&.01  \\
\hline
H & SCAD1 & .07 &.72& 0 &.48&.95&0 &.90&1&0\\
\hline
H & SCAD2&  .35&.72& .30 &.49&98&.01 &.90&1&0 \\
\hline
H &MCP & .33&.70&.29   &.48&.96&.01 & .90&1&0  \\
\hline 
\end{tabular}
\caption{ Variable selection performance for $n=150$}
\label{Table_n=150}
\end{table}

For light tails and low sparsity, Table \ref{Table_n=150} shows that AL and AEN 
are the best methods in terms of variable selection.   For medium sparsity they
remain the best but MCP and SCAD2 are almost as good.  For high sparsity,
AL, AEN, ASCAD1, SCAD1, SCAD2, and MCP are essentially the same.  Note that
all methods have zero in the Fa column meaning they never exclude variables
that are important  For heavy tails,  the results for low sparsity are the
same but for medium sparsity more methods are nearly as good as the AL and AEN.
For high sparsity,  the only methods that are lagging in performance
are L, EN, and L-L; the others are essentially equivalent.

Compared to Table \ref{Table_n=75} we see that all methods improved in
variable selection, which is not surprising, but that the adaptive
methods improved more.   This is true for the instability curves as well.
LM's and RR  are not included in Table \ref{Table_n=150} because they don't 
do variable selection.

Comparing the conclusions from Fig \ref{Fig_n=150} and Table \ref{Table_n=150},
we see that SCAD1 performs best predictively although not necessarily in variable selection.
In Table \ref{Table_n=150}, SCAD1 never excludes a variable that should be included.
Thus, its lesser performance in terms of variable selection may be ameliorated
by its parameter estimation.  

Again, we considered two dependence cases with light tails, the tridiagonal and the
Toeplitz.  For the tridiagonal case, the instability curves and the variable selection table
are qualitatively the same as for the independence case.  For the Toeplitz, 
LM's, SCAD1, SCAD2, and MCP are the only methods that perform well in terms of
instability curves.  In the corresponding table,  overall SCAD2 and MCP did best, but
sometimes another method does best (but then does poorly in terms of instability).
Overall, variable selection in this case is worse than in the independence case.
This is virtually the opposite of Toeplitz in the $n<p$ case where variable selection
was improved.  This case is more in line with intuition because intuition corresponds 
to $n>p$.

\subsection{Sample Size $n=500$}
\label{n=500}
For completeness we also consider the case $n=500$ to identify the limiting behavior
of the methods.  Fig. \ref{Fig_n=500} and Table \ref{Table_n=500} have the same
general properties as the earlier Figures and Tables, namely, as sparsity increases instability decreases.   The methods improve
as sparsity increases although the improvement is not as 
dramatic as in the smaller sample cases. In the heavy tailed cases, there is more variability.   

In fact, many of the methods at this point are indistinguishable via Fig. \ref{Fig_n=500}
or Table \ref{Table_n=500}.    So for descriptive purposes it is easier to identify the
methods that perform poorly rather than the ones that perform well.
From the top row in Fig. \ref{Fig_n=500}, the only poor methods are RR, ASCAD1,L-L
for low and medium sparsity.  For high sparsity RR is clearly worst.
From the bottom row in Fig. \ref{Fig_n=500}, RR is clearly the worst in all cases.
For low sparsity LM and SCAD1 (and nearly SCAD2) are best.  For medium sparsity
only L-L (and RR) performs poorly.   For high sparsity the worst performers are
RR and LM's.    Note that L and EN still perform well even though they don't have
the OP.  However, EN is a generalization of L and L has some consistency properties,
see \cite{Zhao:Yu:2006}, so the good performance of L and EN is not surprising.

\begin{figure}[htp!]
\begin{center}
\begin{tabular}{ccc}
%\subfloat[scatter plot of YIELD vs. TKWT]{
\includegraphics[width=.35\columnwidth]{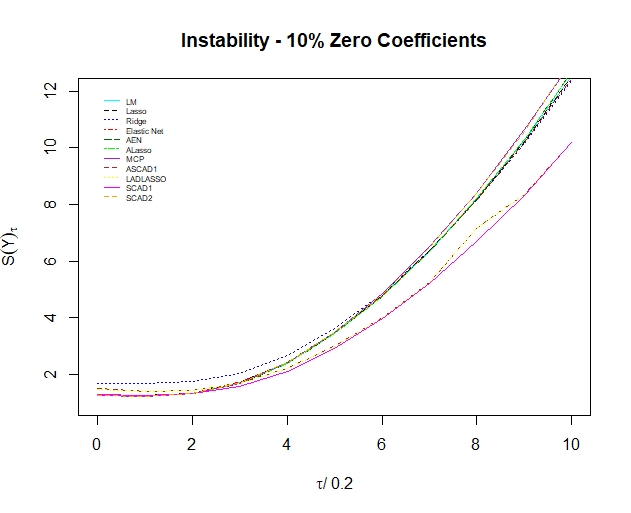}
%\hfill
%\subfloat[scatter plot of YIELD vs. SPSM]{
\includegraphics[width=.35\columnwidth]{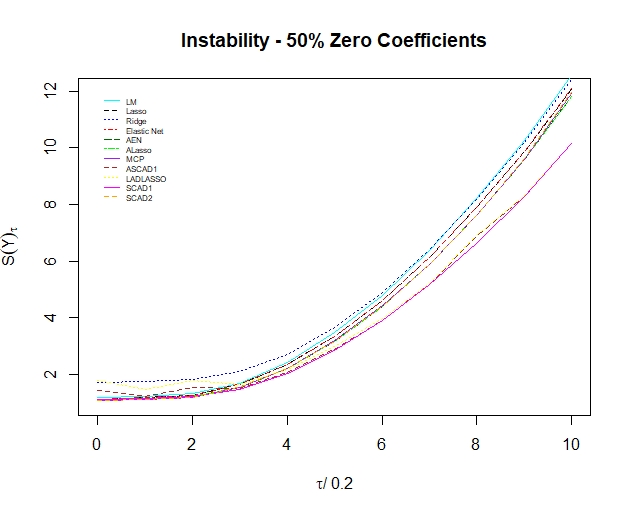}
%\hfill
%\subfloat[scatter plot of YIELD vs. KPS]{
\includegraphics[width=.35\columnwidth]{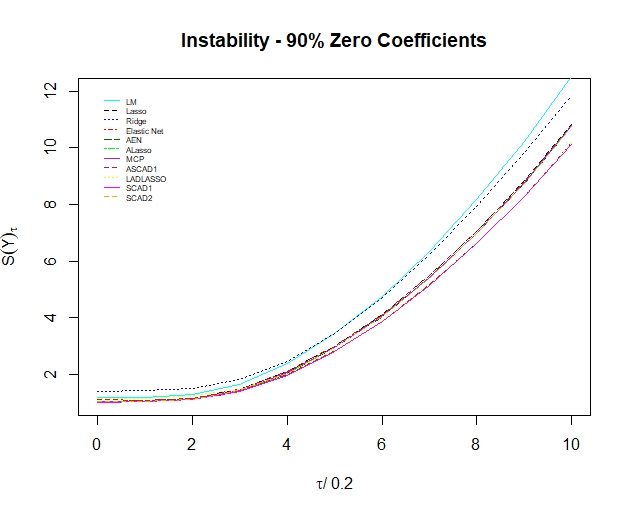} 
\\
%\subfloat[scatter plot of YIELD vs. TKWT]{
\includegraphics[width=.35\columnwidth]{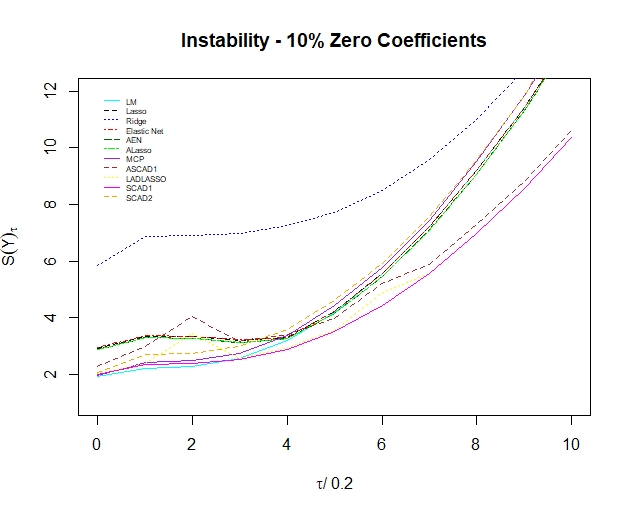}
%\hfill
%\subfloat[scatter plot of YIELD vs. SPSM]{
\includegraphics[width=.35\columnwidth]{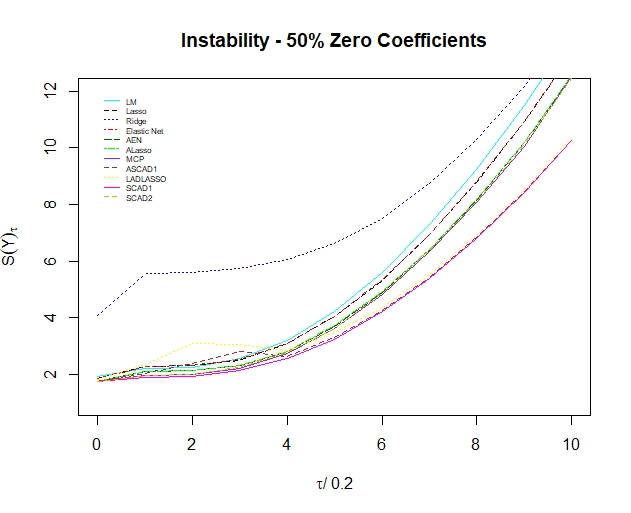}
%\hfill
%\subfloat[scatter plot of YIELD vs. KPS]{
\includegraphics[width=.35\columnwidth]{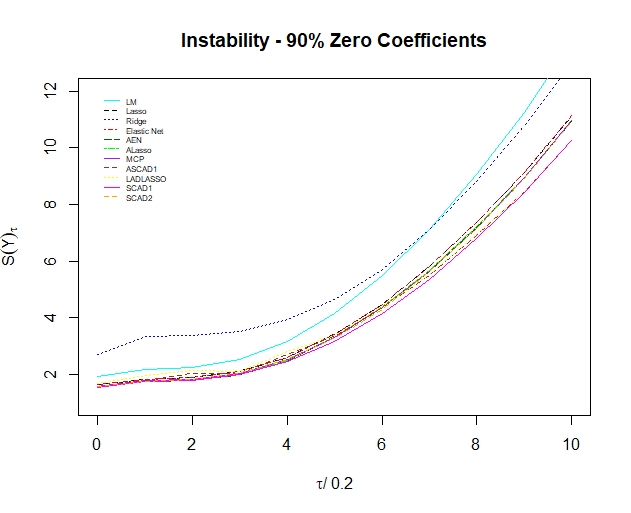}
\end{tabular}
\end{center}
\caption{$n=500$: From left to right the sparsity increases from .1, to .5, to .9.  From top to bottom
the heaviness of the tails of $X$ and $\epsilon$ increases from normal to $t_3$.  The distribution
of the 100 IID outcomes for the parameter $\beta$ is $N(4,1)$.  We used the identity matrix
in the normal distribution for $X$.}
\label{Fig_n=500}
\end{figure}

Table \ref{Table_n=500} shows that at $n=500$, most methods are performing variable selection quite well. In fact, almost all the methods used that have the OP are nearly perfect, on average, in performing variable selection. The two exceptions are LL and ASCAD1 which tend to 
include too many variables. That said, these two methods do not often exclude variables that are important, which allows them to perform comparably well predictively.  Even the methods that look worse in terms of variable selection (L, EN) predict well because they retain all the important variables. 
The only methods this table suggest could be ruled out are in the low or medium
sparsity case:  L, EN,  and L-L. have low entries in the Tot and Tr columns.    (They also have
near zero entries in the Fa column.)
However,  our comparison is predictive not based
on variable selection and these methods perform generally well in a predictive
stability sense.    It is likely that the coefficients of the incorrectly included variables are
quite small.

Note that in this case, the errors at zero of the light tailed methods approach the theoretical lower limit of the predictive error ($\sigma=1$
).  For the heavier tailed methods,
the smallest errors are below two.  Since the error terms here are $t_3$ with 
$\sigma =\sqrt{3} \approx 1.73$,
again the smallest errors are approaching their theoretical lower limit.   (This is seen in
earlier cases but usually only for the highest sparsity.)

 \begin{table}
 \centering
\begin{tabular}{llllllllllll}
\hline
  \multirow{1}{*}  Sparsity & 
\multicolumn{3}{c}{.1}&
\multicolumn{3}{c}{.5}&
\multicolumn{3}{c}{.9}& \\
      &       & Tot & Tr & Fa & Tot & Tr & Fa & Tot & Tr & Fa \\
    \hline
    Li & L  &   .04 & .41 & 0 & .34 & .68 & 0  & .83 & .93 & 0\\
    \hline
    Li & AL & .10 & 1 & 0 & .50 & 1 & 0  & .90 & 1 & 0\\
    \hline
    Li & EN &.04 & .36 & 0 & .34 & .67 &0 & .83 & .92 & 0 \\
    \hline
Li &AEN & .10 & 1 & 0  & .50 & 1 & 0 & .90 & 1 & 0 \\
\hline
Li & ASCAD1 & .09 & .94 & .05 & .50 & 1 & 0 & .91& 1 & .06\\
\hline
Li & L-L & .02 & .15 & .1  & .36 & .72 & .01  &.90 & 1 & .06\\
\hline 
Li & SCAD1 & .09 & .93 & 0 & .48 & .97 & 0 & .90 & 1 & 0 \\
\hline
Li &SCAD2 &  .10 & 1 &0  & .50 & 1 & 0 & .90 & 1 & 0 \\
\hline 
Li & MCP &  .10 & 1 & 0  & .50 & 1 & 0  & .90 & 1 & 0\\
\hline 
\hline
H & L &  .06 & .60 & 0 & .28 & .56 & 0 &  .86 & .95 & 0\\
\hline 
H & AL &  .10 & 1 & 0 & .50 & 1 & 0 & .90 & 1 & 0 \\
\hline 
H & EN & .06 & .56 & 0 & .27 & .55 & 0 & .86 & .95 & 0    \\
\hline
H & AEN &  .10 & 1 & 0  & .50 & 1 &0  & .90 & 1 & 0\\
\hline
H & ASCAD1 & .09 & .96 & 0 & .49 & .97 & 0 & .90 &1 & 0\\
\hline
H & L-L & .01 & .11 & 0 & .06 & .11 & 0  & .86 & .96 & 0 \\
\hline
H & SCAD1 & .09 & .91 & 0 & .49 & .99  & 0 & .90 & 1 & 0\\
\hline
H & SCAD2& .10 & 1& 0 & .50 & 1 & 0  & .90 & 1 & 0\\
\hline
H &MCP &  .10 & 1 & 0  & .50 & 1 & 0 & .90 & 1 & 0 \\
\hline 
\end{tabular}
\caption{ Variable selection performance for $n=500$}
\label{Table_n=500}
\end{table}

\subsection{Summary Tables and Recommendations}
\label{Summarytables}

To specify our recommendations we assume the following, elaborated from
Subec. \ref{n=40}. First, when $p>n$, the instability graphs are more important that the tables because the variable selection is so poor for all methods that we must 
rely solely on choosing the method that predicts the best (or least bad). Thus, for $n=40$ we just consider the figures in choosing the best methods.  For $n=75$, unless one specifies a relative weighting among Tot, Tr, and Fa, there is usually no clearly preferred method.
In these cases, we have used the instability curves to select from amongst the numerically best shrinkage methods.

Second, for larger sample sizes, and low to medium sparsity cases, we use 
both the instability curves and variable selection tables to make recommendations. That is, we choose the methods that are both selecting variables most appropriately and have the lowest instability generally.  Again, there are 
cases where we have to weight Tot, Tr, and Fa.  We continue to believe that Fa is
relatively more important than Tot and Tr and that Tot is relatively more important than Tr.

Two limitations of our recommendations are
that in practice i) the level of sparsity and ii) the variable
selection tables cannot be known.  Accordingly, it is only
the instability curve that is available.    Thus, our recommendations are
based primarily on the instability curves but taking into account the variable
selection tables. Moreover, there are other limitations that a potential user should take into
account e.g.,  our restriction to linear models.  

We begin our recommendations for small sample sizes with Table \ref{All cases-n=40} (based on our computations from Subsec. \ref{n=40}).  Recall that for $n=40$ no method worked particularly well.  
However, the overall best performing of these poor methods was RR.

 \begin{table}%[H]
 \centering
%\centering
%\resizebox{10cm}{!}{
\begin{tabular}[t]{llll}
\hline
 Sparsity & .1 & .5  &  .9  \\
\hline
Ind, Li       & RR/EN  & RR/EN  &   RR/EN      \\
\hline
Ind, H         & RR/EN&  RR/EN &    RR/EN         \\
\hline
Tri  & RR/EN &  RR/EN & RR/EN  \\
\hline 
Toe &RR & RR & RR/SCAD2/MCP\\
\hline 
\end{tabular}
%}
\caption{ Best performing shrinkage methods for $n=40$.}
\label{All cases-n=40}
\end{table}

Recommendations for the next sample size, $n=75$,  are given in Table \ref{All cases-n=75} (based on the computations in Subsec. \ref{n=75}).
EN was typically the preferred method for low to medium sparsity. For high sparsity, all methods besides RR performed well. Interestingly, we found SCAD2 or MCP performed best
when there is 
strong dependence as in the Toeplitz structure.

 \begin{table}%[H]
 \centering

%\resizebox{10cm}{!}{
\begin{tabular}[t]{llll}
\hline
 Sparsity & .1 & .5  &  .9  \\
\hline
Ind, Li    &   EN &  EN &  Not RR   \\
\hline
Ind, H   & EN      &   EN & Not: RR         \\
\hline
Tri  & EN &  EN & Not: RR \\
\hline 
Toe &SCAD2/MCP & SCAD2/MCP & SCAD2/MCP \\
\hline 
\end{tabular}
%}
\caption{ Best performing shrinkage methods for $n=75$.}
\label{All cases-n=75}
\end{table}

Table \ref{All cases-n=150} summarizes our recommendations for $n=150$ (based on the computations in Subsec. \ref{n=150}).
For low sparsity, LM's were generally good.  For medium sparsity,  the conclusions
did not follow a strong pattern.   Several methods performed roughly equally well
and tridiagonal was, as in the earlier cases, similar to the independence cases,
except here for medium sparsity.
With stronger dependence,  the biggest change from the other cases was that SCAD1/SCAD2/MCP
performed best like the medium sparsity case for tridiagonal.
We also see that generally, as sparsity increases, LM's do relatively worse.

 \begin{table}%[H]
\centering
%\resizebox{10cm}{!}{
\begin{tabular}[t]{llll}
\hline
 Sparsity & .1 & .5  &  .9  \\
\hline
Ind, Li       &  LM  &  Not: RR/LL/ASCAD1 &   Not: RR/LM   \\
\hline
Ind, H   &   LM   &  SCAD1/AL/AEN  &   Not :  RR/LM/LL/ASCAD1     \\
\hline
Tri  & LM/SCAD1& SCAD1/SCAD2/ MCP & Not RR/LM  \\
\hline
Toe &LM  & LM/SCAD1 & SCAD1/ SCAD2/ MCP \\
\hline
\end{tabular}
%}
\caption{Best performing shrinkage methods for $n=150$.}
\label{All cases-n=150}
\end{table}

Table \ref{All cases-n=500} summarizes our recommendations for $n=500$ (based on the computations in Subsec. \ref{n=500}) where it is safe to
assume the asymptotic properties have kicked in.
Thus, we see several methods performing well regardless of whether they have the OP.
Indeed, it is very hard to chose a single best method.
There are, however, several methods that essentially never do well, namely, RR, LL, and ASCAD1. 
Further LM also does relatively poorly as sparsity increases.  As such, we recommend 
not using these methods.  Finally, we see from the two dependence cases
that dependence has a very small effect at most when the sample size is this large. 
Indeed, it is the low sparsity with heavy tails that stands out somewhat from the rest.

 \begin{table}%[H]
\centering
%\resizebox{10cm}{!}{
\begin{tabular}[t]{llll}
\hline
 Sparsity & .1 & .5  &  .9  \\
\hline
Ind, Li       & Not: RR/LL/ASCAD1  &Not: RR/LL/ASCAD1   &  Not: RR/LM    \\
\hline
Ind, H   &   LM/SCAD1/SCAD2/MCP   & Not: RR/LL/ASCAD1   & Not: RR/LM        \\
\hline
Tri  & Not: RR/LL/ASCAD1  &Not: RR/LL/ASCAD1   &  Not: RR/LM   \\
\hline
Toe &Not: RR/LL/ASCAD1  &Not: RR/LL/ASCAD1   &  Not: RR/LM   \\
\hline
\end{tabular}
%}
\caption{Best performing of shrinkage methods for $n=500$.}
\label{All cases-n=500}
\end{table}

We conclude with some general observations. In the  $n>p$ case, we observe LM's often perform best except for high sparsity.  
Fortunately, at this sparsity level all shrinkage methods except for RR are essentially equivalent.  LM's are 
sometimes are not very good in the heavy tailed case but, again,  this mainly
occurs  for high sparsity as shrinkage methods set zero coefficients
to zero faster than LM does.  

We observe that the higher the sparsity, the better the methods exclude variables
that are not relevant and include only the relevant variables. 
This is generally true regardless of the sample size.   That is,
we are observing a sort of consistency under increasing sparsity that 
occurs even for sample sizes that are not asymptotic.   

That said, as sample size increases, methods with the oracle property do indeed 
emerge as best, if not uniquely so, regardless of the heaviness of the tails -- as long 
regularity conditions e.g., conditions 2, 3 and 4 in Subsec. \ref{cond2} on moments, are satisfied.

\section{Corroboration on Real Data}
\label{realdata}

As a test of our recommendations in Subsec \ref{Summarytables}, we used the same
shrinkage methods on the data set {\sf Superconductivity} presented in \cite{Hamidieh:2018}.    
This data set has
81 explanatory variables of a physical or chemical nature to explain a response $Y$
representing temperature measurements (in degrees K) for when a compound 
begins to exhibit superconductivity.
Initial data analyses suggested the data were sparse,  but it was unclear how sparse.
\cite{Hamidieh:2018} suggests a sparsity level of about 90\% and our techniques here confirm this in the sense that we find, if a linear model is fit, around 90\% of the coefficients will be zero and this will be nearly best possible from a predictive standpoint.
In addition, \cite{Hamidieh:2018} implicitly used light tails in the error term, $\epsilon$, and
did not comment on the distribution of the explanatory variables apart from effectively
taking them as independent and not requiring any special treatment to account
for spread.  Accordingly, we treated these as coming from a light tailed distribution.

Furthermore, \cite{Hamidieh:2018} identified 20 variables of potential
importance.  Of those 20, we identify only seven of them being important
because the variable importance factors decreased suddenly at the eighth 
most important variable.
This gives 7/81 $<$ 10\%, confirming this case corresponds to the high sparsity setting. 
Histograms of the residuals from the full LM suggest this falls into the light tail case as well. 
Thus, we compare our computed results in this section
to the recommendations for the light tailed high sparsity cases treated
in Subsec.  \ref{Summarytables}.

In fact,  the full {\sf Superconductivity} data set had $n= 21263$, so \cite{Hamidieh:2018} 
was able to use a standard (unpenalized) LM as a `benchmark model' and then 
improved on it by developing an XGBoosting model -- a boosted, penalized tree model
in which the penalty was carefully constructed to be appropriate for trees.

Here, as is common in pratice, especially where a more justifiable methodology is infeasible, 
we have used LM's for their interpretability,  Also, when $n<<p$, XGBoosting 
often does not perform well.   So, it may sometimes be reasonable to use shrinkage 
techniques in mis-specified model situations with small sample sizes.

Since {\sf Superconductivity} is so much larger than the data sets used in our
simulations,  we drew 40, 75, 150,  and 500 data points at random so comparisons with
our recommendations here would be fair.
We note that many data sets are much smaller than {\sf Superconductivity}
so our example here is intended to be suggestive for them too.

We repeated the analyses presented in Sec. \ref{results1} for the independent cases
with light tails but replaced the simulated data with the randomly chosen subsets of {\sf Superconductivity.}  We were able to generate instability curves but not the
variable selection accuracy tables because the true model
is unknown.   The instability curves for {\sf Superconductivity} are given in
Fig. \ref{superconductivity-instability}.

\begin{figure}[h!]
\begin{center}
\begin{tabular}{cc}
\includegraphics[width=.5\columnwidth]{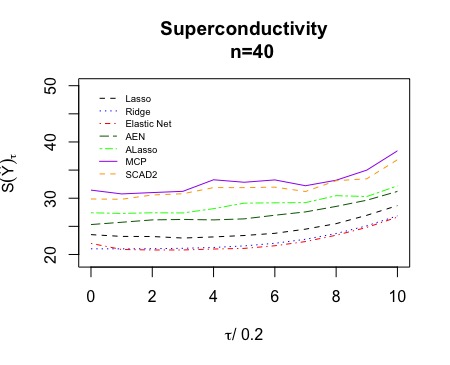}
&
\includegraphics[width=.5\columnwidth]{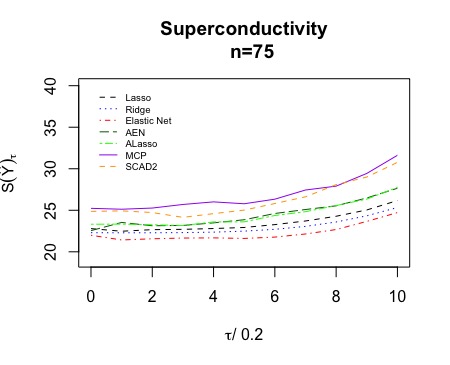}
\\
\includegraphics[width=.5\columnwidth]{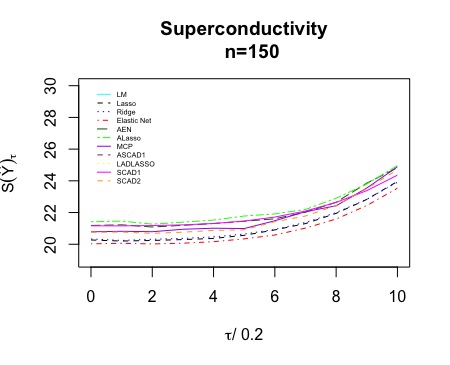} 
&
\includegraphics[width=.5\columnwidth]{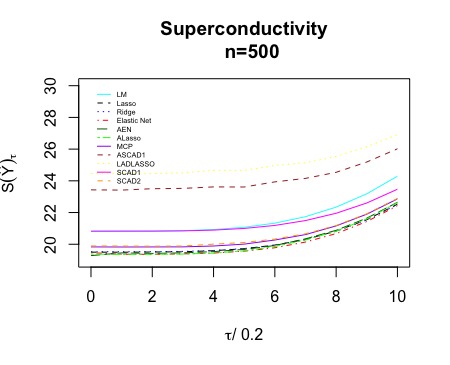} 
\end{tabular}
\end{center}
\caption{Instability curves for the Superconductivity data for $n= 40, 75$ (top) and $n= 150, 500$
(bottom)}
\label{superconductivity-instability}
\end{figure}

For each sample size, we compare the best methods from 
Fig.  \ref{superconductivity-instability} to the corresponding 
recommendations in Subsec.  \ref{Summarytables}.
For $n=40$,  the upper left panel in Fig. \ref{superconductivity-instability} 
shows that RR and EN are the best shrinkage methods.   This is the same as
recommended in Table \ref{All cases-n=40} for sparsity .9 and light, independent tails.
For $n=75$,  the upper right panel in Fig. \ref{superconductivity-instability} shows
EN is the best, followed by RR and L which are noticeably worse.   
Table \ref{All cases-n=75} shows only that RR should not be used with 
 light, independent tails. So, again we see agreement even if the recommendation
was not specific.  

By contrast,  for $n=150$, the lower left panel in Fig. \ref{superconductivity-instability} 
shows that EN is best, closely followed by L and RR.   Table \ref{All cases-n=150} 
indicates that RR and LM's are to be avoided (for light independent tails).
So,the good performance of RR disagrees our recommendations.
Finally, for $n=500$,  the lower right panel in Fig. \ref{superconductivity-instability} 
shows that five methods form a cluster of the best of the 11 methods.  
The cluster of top methods is ALASSO, L,  AEN,  and RR.  The
recommendation from Table \ref{All cases-n=500} is not to use RR or LM's.
Again, we have a disgreement on the use of RR.

We explain these findings by model mis-specification.  First, the true  model is
almost certainly not a LM.  Indeed, \cite{Hamidieh:2018} ends up proposing
a model based on trees.   The agreement between our recommendations and
the data analysis for small values of $n$ probably means that the sample size is too small to
detect the difference between the true model and a LM.  However,
when the sample size increases, the model mis-specification matters.
RR normally performs well for non-sparse cases but here is performing well
when the true model is sparse.   

We conjecture this occurs because a non-sparse linear 
model may provide a better approximation to a sparse non-linear model than a sparse
linear model does.  The analogy is to imagine representing a single true
tree model with a single linear model.  The linear model would have to have many
terms to approximate a tree even one with relatively few nodes.  That is,  a large enough LM 
might provide a good approximation.

As a final point for this section, if we were to redo our simulations using a different model class,
i.e., not LM's, we could end up with different recommendations and if the model class 
contained the true model for the {\sf Superconductivity} data we would expect our
recommendations to match the data analyes.

\section{Optimizing Over the Shrinkage Method}
\label{GAstuff}

Although previous sections used existing well-studied shrinkage methods, the results of Sec. \ref{theory1}
show that there are infinitely many other penalties that could be used to get shrinkage methods
with the OP.   Recalling that penalties are special cases of priors, 
it is clear that the recommendations for choice of shrinkage method given in Sec. \ref{results1}
are limited.   Here we propose that, rather than choosing a shrinkage method from a list
of options, one should find a prior by optimizing a predictive optimality criterion using an adaptive search technique such as a genetic
algorithm (GA) or, more exactly, find a best posterior based on a portion 
of the data that can be used as a prior for remaining and forthcoming data.    

We begin by extending the results of Sec. \ref{theory1} to allow for
data-dependent shifts in parameter locations. The main benefit of shifting the location of the penalty is that it reduces prior-data
conflict.  That way,  when we find an optimal penalty in the next subsection, it will
correspond to putting priors on the $\beta_j$'s that have more of their mass close 
to the true values of the parameters, thereby improving inference.
If the location shift is not used in the penalty, our method below still can be used
but is not as effective. especially in small samples.  It just makes sense that if
the true value of a $\beta_j$ is not zero, then we should not use a prior centered at zero. Then 
we present our GA methodology and verify that our methodology seems to achieve optimality in simulations.  This shows that
a GA approach is a viable alternative to pre-selecting a shrinkage method.

\subsection{Extending the Theory of Sec. \ref{theory1}}
\label{theory2}

Let $\hat{\beta}^*$  be a $\sqrt{n}-$consistent estimator of $\beta$.   
To take advantage of the fact that shrinkage methods can set $\hat{\beta}_j$'s to zero, 
it is natural to choose $\hat{\beta}^*$ to be from a specific shrinkage method such as
SCAD2 that only requires the estimation of one extra parameter.  Adaptive methods 
such as ALASSO, AEN, etc., are also viable.
The idea is to use the $\hat{\beta}^*_j$'s in $\hat{\beta}^*$ 
to adjust the location of the penalty function. 
This overuse of the data is standard in shrinkage methods (in the James-Stein sense)
where it is used to improve the risk performance of decisions.

We state our extensions to Theorems \ref{penllOP} and \ref{empriskOP} as corollaries
since we assume the same hypotheses.
For penalized log likelihoods with location shifted penalties we have the following
sufficient conditions for the OP to hold.

\begin{cor}
\label{cor_ll_OP}
Redefine the objective function in Subsec. \ref{genloglike} to be
$$
Q(\beta) = L(\beta | x^n) + n \sum^p_{j=1} \lambda_j f_j(\beta_j -\hat{\beta}^*_j) .
$$  
Then, under the same conditions as in Theorem \ref{penllOP},  
the estimator $\hat{\beta}=(\hat{\beta_{1}}',\hat{\beta_2}')'$ 
that minimizes $Q(\beta)$ has the OP, i.e.,
$$
P(\hat{\beta_2}=0) \rightarrow 1 \quad {\rm and} \quad \sqrt{n}\hat{I}_{1}(\beta_{1}| x^n) (\hat{\beta_{1}} - \beta_{1}) \rightarrow \mathcal{N}(0 , I_1(\beta_{1}|x^n)).
$$
\end{cor}
\begin{proof}
The proof of Cor. \ref{cor_ll_OP} follows directly from the proof of Theorem \ref{penllOP}.
Indeed, if a true $\beta_j =0$, then for large $n$ we have $P(\hat{\beta}^*_j = 0 ) \rightarrow 1$.
So, in the proof of Lemma \ref{sparsity},the inequalities at the end are asymptotically unchanged.
Then,  since $\hat{\beta}^*_j$ only appears in the penalty, not the likelihood, and in the proof of
Theorem \ref{penllOP} the penalty only has to be controlled in the last term 
of \eqref{asymp_normal},  to get the asymptotic normality
it is enough for $\sqrt{n} \lambda_j \rightarrow 0$, as guaranteed by the hypotheses.
\end{proof}

For penalized empirical risks with location shifted penalties we have the following
analog to Theorem \ref{empriskOP}.
\begin{cor}
\label{cor_emprisk_OP}
Redefine the objective function in Subsec. \ref{penemprisk} to be
$$
Q(\beta) = R(\beta | x^n) + n \sum^p_{j=1} \lambda_j f_j(\beta_j - \hat{\beta}^*_j).
$$  
Then, under the same conditions as Theorem \ref{empriskOP}, 
the estimator $\hat{\beta}=(\hat{\beta_{1}}',\hat{\beta_2}')'$ that minimizes $Q(\beta)$ 
has the OP, i.e., 
$$
P(\hat{\beta_2}=0) \rightarrow 1 \quad {\rm and} \quad 
\sqrt{n}\hat{I^*}_{1}(\beta_{1}| x^n) (\hat{\beta_{1}} - \beta_{1}) \rightarrow \mathcal{N}(0 , I^*_1(\beta_{1}|x^n)).
$$
\end{cor}

Cor. \ref{cor_emprisk_OP} follows from Theorem \ref{empriskOP} the same as
Cor. \ref{cor_ll_OP} follows from Theorem \ref{penllOP}.

The methods motivated by these corollaries continue to allow shrinkage via the $w_j$'s
as well as `James-Stein' type shrinkage.
So, we are introducing another $p$ hyperparameters.   For this reason we only recommend
this `double shrinkage' approach when $n$ is not too much smaller than $p$ and
preferably $n > p$.  For these cases,  we continue to set $\hat{w}_j =1/| \hat{\beta}_{j, OLS}|$.
Taken together, double shrinkage lets us set coefficients to zero from the OP on the
$\hat{\beta}^*_j$'s and from the OP on the $\beta_j$'s.  Moreover, the
weights $\hat{w}_j$ also function to give tighter intervals around $\beta_j$'s that have 
smaller $|\hat{\beta}_{j, OLS}|$'s.    Thus, overall, we do tend to get penalties/priors that are 
centered around zero when they should be and not centered around zero when they shouldn't be.

\subsection{Using GAs to Find a Shrinkage Method}
\label{usingGAs}

Our goal is to find the penalty that leads to the best predictor.   To this end,
recall a GA is a computational algorithm that tries to mimic evolution to optimize a
fitness function by using analogs of mutation, crossover, and selection. 
 Here,
we embed a GA in a two stage optimization to find an optimal penalty/prior.
We start with a `population' of penalty functions and then use gradient descent 
on the $Q$ from each of them to find a corresponding set of $\hat{\beta}$'s.  
We evaluate a fitness function for each $\hat{\beta}$ and select the vectors of $\beta$'s from the current population that correspond to the top 20\% (say) of fitness values as our `elite' set.  Now we use mutation and crossover 
on the remaining penalty functions (i.e., non-elite members of the current population).  The resulting
penalty functions along with the `elite' set then form the next population.  Iterating this process gives
a series of populations of penalty functions with a non-decreasing `best' fitness value and
the $\hat{\beta}$ from the `best' penalty optimizes the fitness function.

First, we define our initial class of penalty functions.
Cors.  \ref{cor_ll_OP} and  \ref{cor_emprisk_OP} imply we can use any penalty function
within a very large class, as long as the regularity conditions are met. 
Since our goal is to find the optimal penalty and the optimal penalty may be very
complicated, we satisfy ourselves with merely approximating it.  
Here we represent $f_j(\beta_j)$ using finitely many polynomials.  That is,
with mild abuse of notation, we set
\begin{eqnarray}
f_j(\beta_j) =  \sum^p_{j=1} \sum^6_{k=1} \alpha_k |\beta_j - \hat{\beta}^*_j|^k.
\label{pen1}
\end{eqnarray}
Obviously, we would get a better approximation to an optimal penalty if we
used more terms but for present purposes sixth order polynomials turned out
to be sufficient.  Our initial population of penalty functions is generated from
\eqref{pen1} by selecting $M$ values of $\alpha = (\alpha_1, \ldots , \alpha_6)$
IID from a Unif[0,20], say $\alpha_m = (\alpha_{1,m}, \ldots , \alpha_{6,m})$ for
$m=1, \ldots , M$.   The GA will update this initial population denoted $A^0=\{ \alpha_1^0, \ldots, 
\alpha_M^0\}$ of size $M$ over $F$ iterations to a final
population $A^F =\{ \alpha_1^F, \ldots, \alpha_M^F\}$
also of size $M$ in which we expect essentially all members to be the same. 
(\cite{Givens:Hoeting:2013} p.  75 states that the algorithm often stops
when there is little diversity in the population, as we detected.)

We start by showing how the typical iteration from $A^0$ to $A^1$ proceeds.
Assume we have data ${\cal{D}} = {\cal{D}}_n  = \{(y_i, x_i) | i=1, \ldots, n \}$ 
and $\dim(x_i) = p$ and the empirical risk
$$
R(\beta | {\cal{D}}_n)  = \frac{1}{n}\sum^n_{i=1}(y_i - x'_i\beta)^2.
$$
In view of Cor. \ref{cor_emprisk_OP} we seek
\begin{equation}
\label{eq_GA1}
\hat{\beta}_{\alpha^0_m} =
\arg \inf_{\beta}  \left( 
\frac{1}{n}\sum^n_{i=1}(y_i - x'_i\beta)^2
+ 
\lambda \sum^p_{j=1} w_j \sum^6_{k=1} \alpha_{m,k}^0 |\beta_j - \hat{\beta}^*_j|^k 
\right)
 \end{equation}
for each $\alpha_m^0 \in A^0$.     
We find suitable values of the decay parameter $\lambda \in \mathbb{R}^+$ and the $\hat{\beta}_j^*$'s based on the data
as described shortly.    We will use two versions of \eqref{eq_GA1} depending on 
the relative sizes of $n$ and $p$.  Specifically,  if $p \geq n$ or not too much smaller than
$n$,  we set all $w_j = 1$ and all $\hat{\beta}^*_j=0$.  If $p < n$, we set 
$w_j = 1/ |\hat{\beta}_{j, OLS}|$
as  noted in Subsec.  \ref{theory2}.
We make this choice because when $n<p$ typically our asymptotic results do not apply.
In the case that $p \geq n$ \eqref{eq_GA1} reduces to
\begin{equation}
\label{eq_GA2}
\hat{\beta}_{\alpha^0_m} =
\arg \inf_{\beta}  \left( 
\frac{1}{n}\sum^n_{i=1}(y_i - x'_i\beta)^2
+ 
\lambda \sum^p_{j=1} \sum^6_{k=1} \alpha_{m,k}^0 |\beta_j|^k 
\right).
 \end{equation}

To solve \eqref{eq_GA1} or \eqref{eq_GA2}, we randomly split the data to estimate the
various parameters.
We begin by writing ${\cal{D}} = {\cal{D}}_{train} \cup {\cal{D}}_{test}$. 
We reserve ${\cal{D}}_{test}$ for comapring predictors after the entire GA process is
completed.
Next we split the training data 
$$
{\cal{D}}_{train} = {\cal{D}}_{train, \lambda}  \cup {\cal{D}}_{train, \beta} \cup {\cal{D}}_{train, \alpha}.
$$ 
We use ${\cal{D}}_{train, \lambda}$ to find $\lambda$ and ${\cal{D}}_{train, \beta}$ to find
the $\hat{\beta}_j^*$'s and the $\hat{w}_j$'s ($n>p$).   Since $\hat{\beta}_{\alpha_m^0}$ depends
on $\lambda$ we begin by searching over a list of values $\Lambda$ equally spaced from
$\lambda_{max} = (1/n_{train, \beta}) \max | Y_{train,\beta}^T X_{train,\beta}|$ to $\lambda_{min} = \gamma \lambda_{max}$ for some $0<\gamma<1$.  
(Here, $n_{train, \beta} = \#{\cal{D}}_{train, \beta}$ with corresponding data indicated by
$Y_{train,\beta}$ and $X_{train,\beta}$.)
For each fixed $\alpha^0_m$ and each choice of $\lambda \in \Lambda$, 
we find $\hat{\beta}_{\alpha_m^0 , \lambda}$ from ${\cal{D}}_{train, \beta}$ and choose the 
$\hat{\lambda}^0_m$ that minimizes
 $R( \hat{\beta}_{\alpha_m^0, \lambda} | {\cal{D}}_{train, \lambda})$. 

We find $\hat{\beta}_{\alpha^0_m}$ for each $m$ in \eqref{eq_GA1} by sub-gradient 
descent since $\alpha^0_m$, $\lambda = \hat{\lambda}^0_m$, 
$w_j = \hat{w}_j$ and $\hat{\beta}^*_j$ can be taken as given.
(The $\hat{w}_j$ and $\hat{\beta}^*_j$ should also have sub- and super-scripts $m$ and 0;
we omit these for convenience.)
Recall, the sub-gradient descent algorithm allows for us to have points of 
non-differentiability in the penalty (e.g., a corner as in L or SCAD2), and in cases where 
the penalty is differentiable, the sub-gradient is uniquely defined by the gradient.  
Note that the objective function is constructed to be convex, so we do indeed 
have a minimum. We initialize the gradient descent algorithm at the LASSO solution 
for $n>p$ and at the RR solution for $n<p$.

Now define the fitness function for the GA to be  
\begin{eqnarray}
\label{fitness}
f = \sum_{i \in {\cal{D}}_{train, \alpha} } (y_i - x'_i \hat{\beta}_{\alpha^0_m,  \hat{\lambda}^0_m})^2.
\end{eqnarray}
We evaluate the fitness for each $\alpha_m^0$ in $A^0$. 
Note that for each $\alpha_m^0$ for $m =1, \ldots, M$ we get a single best choice for 
$\hat{\lambda}_m^0$ and $\hat{\beta}_{\alpha_m^0, \hat{\lambda}_m^0}$ and hence a
single fitness value.
However, it is possible for different $\alpha_m^0$'s to give exactly the same
$f$-value because its possible $\hat{\beta}_{\alpha_i^0, \hat{\lambda}_m^0} = \hat{\beta}_{\alpha_j^0, \hat{\lambda}_m^0}$ for some $i\neq j$.  Although this would appear to happen with
probability zero,  it is observed on a regular basis.  This arises because different but
similar penalties may lead to the same solution and because computing only
has limited precision.

Next, by elitism we select off the top 20\% of members of $A^0$.   We fill in the `missing' 80\% 
by applying crossover and mutation to the bottom 80\% of fitness values
to obtain a new generation of size $M$ from the algorithm
to go into the second iteration.  Crossing means switching some entries of a genome
$\alpha^\prime_m$ with entries from another $\alpha^\dagger_m$ to generate 
a `new genome'.   This
is done at random keeping only the `child' until the population size $M$ is achieved.
Mutation means adding a perturbation to all members of $\alpha$ (here a random number 
between the user specified maximum and minimum values for each component in $\alpha$). 
Mutation does not change the size of the population, only the specific genomes
already in it.  In this way we get a new population $A^1$ to which we can apply the same procedure.
Then, we can iterate to get $A^2$, $A^3$ and so on until $A^F$ contains little
diversity on the population. 

To see that this is the typical behavior of this sort of GA, 
we use the framework of
\cite{Rudolph:1996}.  First, it is easy to see that as we have set it up here,
the GA is a Markov process.  That is, the probabilistic behavior in moving from time $t$ to 
time $t+1$ depends only on the state at time $t$.  Moreover, this Markov process is
homogeneous in the sense that the transition from time step to time step is the
same for any two adjacent time steps.  Note that the Markov process is `discrete time'
and has a discrete population (leading to distinct crosses)
but the mutation is continuous because of the uniform distribution.  Thus, there is no
transition matrix.  Instead, there is a transition kernel, $K(x, S)$, where $x$ is a
population member at time $t$ and $S$ is a set of possible states to which $x$
may be transformed and $K$ is independent of $t$.    In fact, $K(\cdot, \cdot)$ can be
partitioned into a $K_m$ and $K_c$, a mutation and crossover kernel.
The crossover kernel is a transition matrix since crossover is discrete.
The mutation kernel includes the continuous mutation phase based on the
uniform distribution.  So,  let $x$ be any state at time $t$ and
suppose an optimum 
$f^*$ exists and the Markov process has state space $E$.  Then,  there will be elements of 
$E$ arbitrarily close to $f^*$.   Let $b(x_t)$ be the best fitness value within 
the $t$-th population and let $d(x) = b(x) - f^*$.    As long as the population is large enough,
$B_\epsilon = \{ d(x) < \epsilon\}$ will have nonzero probability for $\epsilon >0$
and hence $K_m(x,  B_\epsilon)$ will be bounded away from zero.
Now,  given that we have used elitism, Theorem 2 in \cite{Rudolph:1996} applies to give 
convergence of the GA to the global minimum of $f$ within the class of priors that
have the OP as in Cors.  \ref{cor_ll_OP} and \ref{cor_emprisk_OP}.

The behavior of the GA depends on $M$,  the elements of $A^0$, the size of $F$, 
the choice of $f$,  the data, etc.   Indeed, a pragmatic check on the behavior of a GA
would be to run it with different initial populations to see if the GA outputs approximately
the same minimum.   To ensure convergence of the GA one should set a large 
population as well as a large number of generations.   We comment that in the
sub-gradient descent phase of our procedure, we have limited ourselves to
convex objective functions.  For more general results we would have to ensure
convergence of the gradient based optimization to ensure convergence of the 
GA-based optimization.   We implemented our GA computations using {\sf genalg},
see \cite{Willighagen:etal:2015}.

Our intuition tells us that this method will be beneficial in low to medium 
sparsity cases, as well as heavy dependence or non-asymptotic cases.
This is due to the fact that asymptotically the OP methods are equivalent, and thus a GA can do no better. Further, in the smaller sample cases with high sparsity we observe most methods performing roughly the same. In low to medium sparsity cases, there is more variability between the methods and thus, we should be able to optimize to find a penalty that in fact does perform better than the common shrinkage methods.  We acknowledge that only using standard basis expansions may not allow us to approximate some penalties well. 

\subsection{Simulations}
\label{results2}

Here we present two simulations, one for $ p >n$ and one for $n>p$
to show how implementing the GA performs relative to other shrinkage methods in a  predictive setting.  We simulate IID observations from 
$$ 
Y = X\beta + \epsilon
$$
where $X \sim MVN_p(0,I_{100})$, $\epsilon= (\epsilon_1, \ldots, \epsilon_n )^T$
with $\epsilon_i \sim N(0,1)$, and $\beta = (\beta_1,  \beta_2)^T$
and we set $p=100$, as before.   We assume 50\%  sparsity, so the dimension of both
$\beta_1 $ and $\beta_2$ are 50.   We take $\beta_1 \sim MVN_{50}(4, I_{50})$
and set $\beta_2 = 0$.  We consider $n=40$ and $n=150$ and
we split the data as described in Sec. \ref{usingGAs}.  

The GA will find an optimal penalty as defined by an optimal vector
$\alpha_{opt} = (\alpha_{1, opt}, \ldots, \alpha_{6, opt})^T$.  The entries
$\alpha_{j, opt} = \alpha_{j, opt}({\cal{D}}_{train, \alpha})$ so we are treating
the penalty as a hyperparameter in the prior that would be mathematically equivalent
to it.   The difference from actually estimating a hyperparameter comes from the
fact we are only using ${\cal{D}}_{train, \alpha} \subsetneqq {\cal{D}}_{train}$.   Given
the penalty, we have a potentially new shrinkage method, dependent on a 
proper subset of ${\cal{D}}_{train}$.
So, we can form a posterior using the prior determined from the penalty 
given the rest of the data.   This posterior can be used
to generate predictions for ${\cal{D}}_{test}$ that can be compared with the predictions
from the other shrinkage methods used in Sec. \ref{results1}.

\subsubsection{GA example $n=40$}
\label{40}

Here we split the data so that 
$\#({\cal{D}}_{train}) = \#({\cal{D}}_{train, \lambda} \cup{\cal{D}}_{train, \beta} \cup {\cal{D}}_{train, \alpha}) = 36$  
with corresponding sample sizes $(2,30,4)$ and $\#({\cal{D}}_{test}) = 4$. 
All other methods we compare use all of the training data to find estimates of 
$\beta$ and $\lambda$.   Note for comparisons with other methods, those that
use {\sf glmnet} use all 36 observations in the training data to form the predictor 
and  for the methods that we implemented with LLA we combined  
${\cal{D}}_{train, \lambda} $ and ${\cal{D}}_{train, \alpha} $ to estimate $\lambda$. 
Thus, we ensured that each method used all the training data, providing a fair comparison.  

Interestingly, but perhaps not surprisingly, we find 
$\hat{\alpha} = (0, 1,0,0,0,0)$ which corresponds exactly to RR. 
This is consistent with the methods that performed the best in 
the analogous cases in Subsec. \ref{n=40}. 
This suggests that when we have few data points relative to 
explanatory variables, we do not have enough information to obtain an informative 
prior (in terms of its location and variance) so we default to the prior 
that makes us retain all the explanatory variables.

The predictive errors for $\#({\cal{D}}_{test}) $ are given in Table \ref{Tab_GA2}. 
Note that GA, RR,EN, and AEN are the same here, because we set $w_j =1, \forall j$. 
We comment that because GA's require a lot of computing 
time, we have not averaged over many data sets to get the prediction errors
reported in this table.  However, we believe we have used a large enough population
and large enough number of generations that our results are accurate.

\begin{table}%[H]
\centering
\resizebox{10cm}{!}{
\begin{tabular}[t]{llllllllllll}
 GA &   L &RR&EN&AEN&AL&SCAD2 &MCP\\
\hline
  22.75 & 45.37    & 22.75  &22.75 &22.75 & 46.62 & 36.18 &36.18    \\
\end{tabular}
}
\caption{MSPE for our new GA method and seven other methods.}
\label{Tab_GA2}
\end{table}

This example illustrates that by optimizing over the choice of 
penalties, we are not guaranteed to find a penalty that is different 
from an established method (although we argue this is the typical case). 
The guarantee is only that we will find an optimal penalty for prediction
and it is no surprise if there are settings where a well known technique
is optimal.   The novelty in our GA approach is that it can be used in any
linear regression problem and, if properly implemented, will always
give the best predictions.

\subsubsection{GA example $n=150$}
\label{150}

When splitting the data in this situation, we must keep more than 100 observations 
in ${\cal{D}}_{train, \beta}$ to ensure $n>p$.  Accordingly, we set
$$
\#({\cal{D}}_{train}) = \#({\cal{D}}_{train, \lambda} \cup{\cal{D}}_{train, \beta} \cup {\cal{D}}_{train, \alpha}) = 135
$$  
with corresponding sample sizes of  $(9,113,13)$ respectively,  and  $\#({\cal{D}}_{test}) = 15$. 
Parallel to our methodology in Subsec. \ref{40},  the methods implemented using
{\sf glmnet} and {\sf rqPen} used all 135 observations in the training data to form 
the predictor.  Also,  as before, for the methods that we implemented with LLA we combined  
${\cal{D}}_{train, \lambda} $ and ${\cal{D}}_{train, \alpha} $ to estimate $\lambda$. 
Again, this ensured all methods were being treated fairly.

After running the GA we found $\hat{\alpha} = (2, 0,0,20,0,0)$.   The associated prediction 
error on ${\cal{D}}_{test}$ for each method is given in Table \ref{Tab_GA1}. 
We observe the penalty selected through GA achieves the best predictive error among 
all methods considered.  As in Subsec. \ref{40},  we comment that because 
GA's require a lot of computing time, we have not averaged over many data sets to 
get the prediction errors reported in this table.  However, we believe we have 
used a large enough population and large enough number of generations that our 
results are accurate.

Since we found a new (and better) penalty, we have graphed it in Fig.  \ref{Fig_GApen}.
Recall that half the parameter values are zero, half are non-zero, and the penalty
term has $f_j$'s, i.e., the terms depend on the index of the parameter.  This means that
we allow different penalties on different parameters.
The left hand panel shows a plot of the optimal penalty for
one of the $\beta_j$'s that is known to be zero.  It is compared with the common penalties 
RR, LASSO, and SCAD.    The right hand panel shows a plot of the optimal penalty for
one of the $\beta_j$'s that is known to be non-zero, again compared with RR, LASSO, and SCAD.
It is obvious that the GA method described in Subsec.  \ref{theory2} gives two
sorts of $f_j$'s.   The training data forces the $f_j$'s corresponding to $\beta_j = 0$
to concentrate at zero and forces the $f_j$'s that correspond to nonzero $\beta_j$'s to
concentrate away from zero.  This explains the improvement in prediction error
seen in Table \ref{Tab_GA1}.

\begin{table}%[H]
\centering
%\resizebox{10cm}{!}{
\begin{tabular}{llllllllllll}
 GA & LM&  L &RR&EN&AEN&AL&ASCAD1&SCAD1&SCAD2&LL    \\
\hline
 1.12&1.52 &  1.78 & 3.66 & 1.78 & 1.50& 1.50& 1.36& 1.53& 1.18& 1.67     \\
\end{tabular}
%}
\caption{MSPE for our new GA method and for 10 other methods.}
\label{Tab_GA1}
\end{table}

To end this section, note that our simulations only show proof of concept;
the priors we found here may not be genuinely optimal for prediction. 
That is so because we have not run the GA for many generations with a large population size
so we cannot assume the GA has converged.
In fact, in both cases here ($n=40,150$) we only ran a single generation of the GA
and we only used a population size of $M=150$.    However, because of the elitism operation,
running the GA longer can never result in a worse predictor and our results
show that it can be relatively easy to find a penalty that is better for prediction 
than established penalties --
even if they are not optimal within the class of all penalties with the OP.

\begin{figure}[h!]
\begin{center}
\begin{tabular}{ccc}
%\subfloat[scatter plot of YIELD vs. TKWT]{
\includegraphics[width=.45\columnwidth]{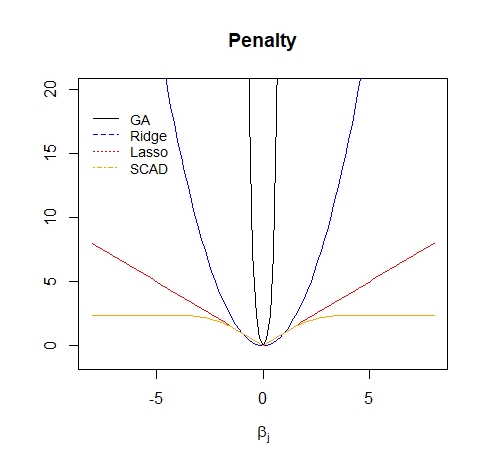}
%\hfill
%\subfloat[scatter plot of YIELD vs. SPSM]{
\includegraphics[width=.45\columnwidth]{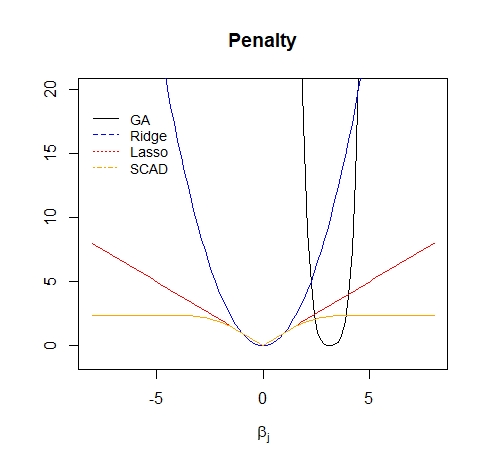}
\end{tabular}
\end{center}
\caption{GA Optimal Penalty vs Standard Penalties for Zero and Nonzero
$\beta_j$'s.}
\label{Fig_GApen}
\end{figure}

\section{Discussion}
\label{conclusions}

This paper has assumed a predictive stability perspective
and within that context shown several results that may be a bit unexpected.
First, the OP is not rare; it is actually rather common.  Its proof requires little
more than what most would regard as regularity conditions.
Second,  for small $n$ and large $p$ shrinkage methods did not perform very well
even if they have the OP and the true model has reasonable sparsity.
Third, on the other hand, if optimal or near optimal penalties are used they give
shrinkage methods that work noticeably better than the established ones.   
Our findings indicate that methods having the OP do not perform
particularly well for $n<p$ and for $n>p$ the OP is no guarantee
that they perform better than methods not having the OP.
Fourth,  our results also indicate that with increasing sparsity
the performance of shrinkage methods improves.   This intuition
needs to be developed further because the obvious limit of perfect
sparsity gives the trivial model.

So, even though the OP is important, it is not at all clear how 
important it is or when it is important in cases where sample sizes are finite. 
We still think it's better to have the OP than not if only because it gives
consistency, asymptotic normality and efficiency.   This is especially
the case with high sparsity and large $n$ relative to $p$, but in these 
cases other methods often perform comparably. 
When $n$ is small compared to $p$, the OP is not a useful property, 
and thus the adaptive methods that have the OP do not perform well. 
A possible explanation for this  is a poor bias variance trade off when $p > n$.
It does not seem to be a good idea to use methods that require estimating $w_j$ for each
$\beta_j$:  For $p>n$ we have not seen any example where the adaptive penalty
gives better results that the nonadaptive version.

Since the OP really requires $n \rightarrow \infty$ whereas $n$ often must be taken as truly
finite, we introduce the notion of instability of predictions as a criterion for selecting a 
penalty or prior.   Comparing instability curves is a finite sample check for 
good predictive performance.  
Using this approach we can easily rule out unstable predictors.
For instance, with high sparsity using a linear model by itself is often unstable.
In general, quantifying the variability of variable selection when $p$ is large is 
difficult, so defining instability in terms of the prediction errors seems reasonable.

Our simulation studies show that as a generality, shrinkage methods tend to perform better in 
terms of variable selection, and thus prediction, as sparsity increases as well as when $n$
increases.  In fact, our simulations showed that regardless of $n$, as the sparsity 
increased, the methods seemed to perform roughly equally well. For instance, recall the 
$n=75$ simulations in Subsec. \ref{n=75} . At $90\%$ sparsity, we observed what appeared to be asymptotic convergence of the method with the OP. Of course this situation is not 
asymptotic as $n < p$, but an increase in sparsity is associated with an increase in efficiency
of the methods. 

Since there are infinitely many such choices for penalties that have the OP, we take 
the subjectivity out of penalty (or prior) selection by using a GA to find an 
optimal penalty/prior for prediction.  When $n>p$ we use the GA approach 
to find a predictively optimal penalty that has the OP. When $n<<p$,
we do not search for methods with the OP because we do not benefit from 
the known asymptotic results.  Thus, we search over the class of penalties that are 
non-adaptive and do not require estimation of many hyperparameters.  
In principle, as long as we let the GA run long enough to converge, this 
approach can never do worse that simply choosing a standard shrinkage method. 
In fact, we have examples where the GA approach does better 
than others; when the GA approach selects the best among standard methods, 
we can infer the standard method was the right choice.

Another way to look at the procedure in Sec. \ref{GAstuff} is that when we find a penalty/prior
based on the data we are producing an approximation
to a predictively optimal posterior given the training data that can then be used with the
log-likelihood.
Thus, the predictive improvement comes from the efficiency of the way the posterior uses
the data with an optimal prior.  That is, we are using the data to form a pseudo-posterior based on 
a de facto optimal prior only determined by the (pseudo-)posterior it forms.

We close with another heuristic that seems to be borne out by our results.
Namely, we associate corners and other points of non-differentiability in the penalty
with setting parameter values equal to zero in finite samples.
Recall that minimizing 
$$ 
\sum^n_{i=1}L(y_i - x'_i \beta) + n \lambda \sum^p_{j=1} w_j f_j(\beta_j) 
$$ 
is equivalent to minimizing $Q =  \sum^n_{i=1}L(y_i - x'_i \beta)$ subject to the constraint 
$\sum^p_{j=1} w_j f_j(\beta_j)   \leq R \in \mathbb{R}$ where $R$ typically decreases
as $\lambda$ increases.
 Denote the constraint region by
$$ 
D =  \left\{ \sum^p_{j=1} w_j f_j(\beta_j)  \leq R \right\}.
$$ 
Since $D$ is closed and compact,
the Krein-Milman theorem (see \cite{Royden:FitzPatrick:2010}) 
gives that $D$ is the closed convex hull of its
extreme points,  i.e., $D = CCH(D_{ext})$.
For reasonable choices of $f_j$, $D$ is defined by the intersection of 
regions of the form 
$$
U_k (\beta_1, \ldots, \beta_p) \leq \alpha_k , k=1, \ldots , p,
$$
where the $U_k$'s are defined from the $f_j$'s and $w_j$'s.
When our goal is optimizing $Q$ over $D$,  it is often the case that
the optima occur at extreme points of $D$.  When
the extreme points of $D$ are on the coordinate axes we will find
at least some of $\beta_j$'s are zero.  Indeed, if $Q$ is convex and 
continuous on an open set containing $D$, then $Q$ often attains its 
minimum over $D$ on a `face 'of $D$
and the exact point where the optima occur may lie at the intersections 
of some or all of the $U_k =\alpha_k$;
this defines a subset of the extreme points of $D$.
This is well-established for the case of linear optimization with linear constraints.
Indeed, if $Q$ is minimized for at least one extreme point of $D$ 
that lies on a coordinate axis then at least some $\beta_j$'s will be set to zero.
This means that any locally convex penalty with a `corner' on a coordinate axis
will perform nontrivial variable selection if 
$R$ is small enough.  We conjecture a converse to this statement will hold, too.

%Taken together, when it comes to choosing penalties or 
%predictors, the data is smarter than we are.

\newpage

%%%%%%%%%%%%%%%%%%%%%%%%%%%%%%%%%%%%%%%%%%%%%%
%% Supplementary Material, if any, should   %%
%% be provided in {supplement} environment  %%
%% with title and short description.        %%
%%%%%%%%%%%%%%%%%%%%%%%%%%%%%%%%%%%%%%%%%%%%%%
%\begin{supplement}
%\stitle{Proofs from Sec.  \ref{theory1}}
%\sdescription{Proofs for the penalized likelihood and empirical risk cases.}

\section{ Supplement : Proofs from Sec. \ref{theory1}}
\subsection{Penalized Log-Likelihood Case, Sec. \ref{genloglike}}
\label{penlltheory}

\begin{thm}
\label{consistency}
 Suppose Conditions 1--6 are satisfied and suppose $\sqrt{n}a_n \rightarrow 0$ where $a_n$ satisfies  $a_n = \frac{1}{h(n)\sqrt{n}}$ in which $\frac{1}{h(n)\sqrt{n}} \rightarrow 0$ and $h(n) \rightarrow \infty$. Then,  there exists a local minimizer $\hat{\beta}$ of $Q(\beta)$ such that $|| \hat{\beta} - \beta_0|| = O_p(n^{-\frac{1}{2}} + a_n)$.
\end{thm}

\begin{proof}
\textbf{Step 1.}
We want to show for any $\epsilon > 0$ there exists a large constant $C$ such that 
\begin{equation}
\label{showthis} 
P\left\{ \inf_{||u|| = C} Q(\beta_0 +\alpha_n u ) > Q(\beta_0)\right \}\geq 1-\epsilon ,
\end{equation}
where $\alpha_n = n^{-\frac{1}{2}} + a_n$.
Denote
\begin{align*}
D_n(u) & = Q(\beta_0 +\alpha_n u ) - Q(\beta_0) \\
 & = L(\beta_0 +\alpha_n u |x^n) + \sum_{j=1}^p \lambda_j f_j( \beta_{j0} + \alpha_n u_j) - L(\beta_0|x^n) - n \sum_{j=1}^p \lambda_j( f_j(\beta_{j0}))\\
 & = L(\beta_0 +\alpha_n u |x^n) - L(\beta_0|x^n) +  n  \sum_{j=1}^p \lambda_j [f_j( \beta_{j0} +\alpha_n u_j) - f_j( \beta_{j0})]  \\
  & = L(\beta_0 +\alpha_n u |x^n) - L(\beta_0|x^n) +  n  \sum_{j=1}^{p} \lambda_j   f_j'(\tilde{\beta_j}) \alpha_n u_j\\
& = L(\beta_0 +\alpha_n u |x^n) - L(\beta_0|x^n) +   n  \sum_{j=1}^{p_0} \lambda_j f_j'(\tilde{\beta_j}) \alpha_n u_j +  n \sum_{j=p_0 +1}^{p} \lambda_j f_j'(\tilde{\beta_j}) \alpha_n u_j 
 \end{align*}
 where $ f_j'(\tilde{\beta_j}) \alpha_n u_j$ is the Taylor expansion of the penalty term by Conditions 4 and 6 and where $\tilde{\beta_j}$ is on the line joining $\beta_j$ to $\beta_j + \alpha_n u$. More formally, $\tilde{\beta_j} \in <\beta_j , \beta_j +\alpha_n u > $. Thus we have,
\begin{align}
\nonumber D_n(u) & = L(\beta_0 +\alpha_n u |x^n) - L(\beta_0|x^n) +  n \alpha_n \left(  \sum_{j=1}^{p_0} \lambda_j f_j'(\tilde{\beta_j})  u_j  +n  \sum_{j=p_0 +1}^{p} \lambda_j f_j'(\tilde{\beta_j})  u_j \right)\\
 \nonumber  &= L(\beta_0 +\alpha_n u |x^n) - L(\beta_0|x^n) +  n (n^{-\frac{1}{2}} + a_n ) \left( \sum_{j=1}^{p_0} \lambda_j f_j'(\tilde{\beta_j})  u_j  +n \sum_{j=p_0 +1}^{p} \lambda_j f_j'(\tilde{\beta_j}) u_j \right) \\
 \nonumber &= L(\beta_0 +\alpha_n u |x^n) - L(\beta_0|x^n) + ( \sqrt{n}  + n a_n ) \left(  \sum_{j=1}^{p_0} \lambda_j f_j'(\tilde{\beta_j})  u_j  +n \sum_{j=p_0 +1}^{p} \lambda_j f_j'(\tilde{\beta_j}) u_j \right)\\
 \nonumber & \geq L(\beta_0 +\alpha_n u |x^n) - L(\beta_0|x^n) + ( \sqrt{n}  + n a_n )a_n \left( \sum_{j=1}^{p_0} f_j'(\tilde{\beta_j})  u_j  +n  \sum_{j=p_0 +1}^{p} f_j'(\tilde{\beta_j})  u_j \right) \\
 \nonumber & = L(\beta_0 +\alpha_n u |x^n) - L(\beta_0|x^n) + ( \sqrt{n} a_n  + n a_n^2 ) \left( \sum_{j=1}^{p_0} f_j'(\tilde{\beta_j})  u_j  +n \sum_{j=p_0 +1}^{p} f_j'(\tilde{\beta_j}) u_j \right) \\
 &= L(\beta_0 +\alpha_n u |x^n) - L(\beta_0|x^n) + \sqrt{n} a_n(1  + \sqrt{n} a_n ) \left(  \sum_{j=1}^{p_0}f_j'(\tilde{\beta_j})  u_j  +n  \sum_{j=p_0 +1}^{p} f_j'(\tilde{\beta_j}) u_j \right)  \label{dnu}
 \end{align}
where $p_0$ is the number of components in $\beta_{10}$.

Now by Taylor expansion of the log likelihood at $\beta_0$, 
\begin{align}
\nonumber L(\beta_0) + \alpha_nu) - L(\beta_0|x^n) & = \alpha_n u L'(\beta_0|x^n) + \frac{n}{2} (\alpha_n u)' \left(I(\beta_0|x^n)\right) (\alpha_n u)  \\
\nonumber &+  \frac{n}{2} (\alpha_n u)'\left(\frac{L''(\tilde{\beta}|x^n}{n} - I(\beta_0|x^n) \pm \frac{1}{n}\sum^n_{i=1} I(\beta_0|x^n) \pm I(\beta_0 )\right) (\alpha_n u) \\ 
&= \alpha_n L'(\beta_0|x^n) u  - \frac{n}{2} u' I_n(\beta_0) u  \alpha^2_n \{1 +o_p(1) \} \label{lnu}
\end{align}
where $I_n(\beta_0)=\frac{1}{n}\sum^n_{i=1} I(\beta_0|x^n)$ by Condition 4 and we have 
$$ 
E\left[\sup_{\beta \in B(\beta_0,\eta)} \left|\frac{L''(\tilde{\beta}|x^n)}{n} - I(\beta_0|x^n)\right| \right]\rightarrow 0 \text{ as } \eta \rightarrow 0 
$$ 
which implies that $I_n(\beta_0) \rightarrow I(\beta_0)$ by Condition 3. Using $(\ref{dnu})$ and 
$(\ref{lnu})$ , 
\begin{eqnarray}
\label{A2}
 \nonumber 
D_n(u) &\geq& \alpha_n L'(\beta_0|x^n)' u  + \frac{n}{2} u' I(\beta_0|x^n) u  \alpha^2_n \{1 +o_p(1) \} + \sqrt{n} a_n  ( 1  + \sqrt{n} a_n)  \sum_{j=1}^{p_0}f_j'(\tilde{\beta_j})  u_j \\
&& +  \sqrt{n} a_n  ( 1  + \sqrt{n} a_n)  \sum_{j=p_0+1}^{p}f_j'(\tilde{\beta_j})  u_j .
\end{eqnarray}

We argue the second term on the RHS of (\ref{A2}) is dominant. To see this consider the first term on the RHS of (\ref{A2}) using Conditions 1 and 3 We multiply by $\frac{\sqrt{n}}{\sqrt{n}}$ and obtain $\sqrt{n} \alpha_n \frac{L'(\beta_0|x^n)}{\sqrt{n}}u$. Let $\bar{L} = \frac{L'(\beta_0|x^n)}{\sqrt{n}}$, and observe $\sqrt{n} \alpha_n \rightarrow 1$. So,  $C^2 = ||u||^2 = u'u > \bar{L} u$ as long as $0 < \frac{1}{B} < \bar{L} <  B$  with high probability for $C >B$.  Then with high probability  $C^2 \geq B \underline{1}(u) \geq B \bar{L} $ for $||u|| =C$ and $\underline{1} = (1,1, \ldots,1)$.  Thus we have that by choosing sufficiently large $C$, the second term of (\ref{A2}) is larger in absolute value, with high probability than the first term uniformly in $||u|| =C$.

Also, by hypothesis $\sqrt{n} a_n \rightarrow 0 $, so the whole third term in (\ref{A2}) goes to zero and the second term of (\ref{A2}) is also larger than the third term. The fourth term of (\ref{A2}) also goes to zero because for $p_0 \leq j \leq p$, we know that $\beta_j =0$, so $\tilde{\beta_j} \rightarrow 0$ as $ n \rightarrow 0 $.Therefore  $  \sum_{j=p_0+1}^{p}f_j'(\tilde{\beta_j})  u_j \rightarrow 0$ by Condition 5 and since $\sqrt{n} a_n \rightarrow 0$, the fourth term on the RHS of (\ref{A2}) goes to 0. Thus since the third term third and fourths term on the RHS of (\ref{A2}) both go to 0, we have that the second term on the RHS of (\ref{A2}) is larger than all of the other terms. 

Therefore,  by choosing large enough $C$, 
$$ 
P\left\{\inf_{||u||=C} D_n(u) \geq  \frac{n}{2} u' I(\beta_0|x^n) u  \alpha^2_n \{1 +o_p(1) \} \right \} \geq 1-\epsilon .
$$ 
Thus, (\ref{showthis}) is true for sufficiently large $C$.

\textbf{Step 2.}
We now show the conclusion of Step 1 holds for any $C^* > C$.
Let $C^\star > C$. Then,
\begin{align*}
\hat{\beta} & \in \{\beta_0 + \alpha_n u : ||u|| =C^* \} \\
&\iff \hat{\beta}  \in \cal{B}\left(\beta_0 , \alpha_n C^*\right) \\
&\iff \hat{\beta} - \beta_0  \in B\left(0, \alpha_n C^* \right) \\
&\iff \alpha_n^{-1}(\hat{\beta} - \beta_0)  \in B(0, C^*) \\
&\iff ||\alpha_n^{-1}(\hat{\beta} - \beta_0)|| < C^*\\
&\iff||\hat{\beta} - \beta_0||^2 < \alpha_n^2{C^*}^2
\end{align*}

It follows that  $ \sum^{p}_{j=1}(\hat{\beta}_j - \beta_{0j})^2 \leq \alpha_n^2{C^*}^2$  and $\forall j$ , $(\hat{\beta}_j - \beta_{0j})^2 \leq \alpha_n^2{C^*}^2$.   So,  
$$
 -\alpha_n C^* \leq \hat{\beta}_j - \beta_{0j} \leq \alpha_n C^*
$$ and we can absorb $C^*$ into $\alpha_n$ because $C^*$ is a constant. Thus we have $\hat{\beta} - \beta_0  = O_p(\alpha_n )$. Note that 
$$
\alpha_n = \frac{1}{\sqrt{n}} + a_n = \frac{1+\sqrt{n}a_n}{\sqrt{n}}
$$
and as $n \rightarrow \infty$ we have $\sqrt{n} a_n \rightarrow 0$, so $\hat{\beta} - \beta_0  = O_p\left(\frac{1}{\sqrt{n}}\right)$. 
\end{proof}

\noindent
{\bf Remark:}
The argument for Theorem 1 is true for differentiable $f_j(\cdot)$, but not non-differentiable $f_j(\cdot)$. In the case of non-differentiable $f_j(\cdot)$, consider a smooth function that approximates $f_j(\cdot)$ well. Say, $f_j^{\star}(\cdot)$, which differs from $f_j(\cdot)$ by a margin of $\epsilon$ where $\epsilon_N \rightarrow 0$ as $n \rightarrow \infty$. As long as $ \lim_{ \beta \rightarrow \beta^{* -}} f_j^{\star}(\cdot) =  \lim_{ \beta \rightarrow \beta^{* +}} f_j^{\star}(\cdot)$, then the above argument holds. 
%\end{remark}

\begin{lemma}
\label{sparsity}
Assume conditions 1- 6, and the result from Theorem \ref{consistency} holds.  If $\sqrt{n}b_n \rightarrow \infty$ and $b_n \geq \frac{g(n)}{\sqrt{n}}$ as $n, g(n) \rightarrow \infty$, then with probability tending to 1, for any given $\beta_1$ satisfying $|| \beta_1 - \beta_{10}|| = O_p(n^{-\frac{1}{2}})$ and any constant $C$, 
\begin{equation}\label{min_beta}
 Q \left \{ \left(\begin{matrix} \beta_1 \\0\end{matrix}\right) \right \} =  \min_{||\beta_2 ||\leq Cn^{-\frac{1}{2}}} Q \left \{ \left(\begin{matrix} \beta_1 \\ \beta_2\end{matrix}\right) \right \}.
 \end{equation}
\end{lemma}

\begin{proof}
Consider the objective function $Q(\beta) = L(\beta|x^n) + n \sum^p_{j=1} \lambda_j f_j(\beta_{j})$. Note that 
\begin{equation}\label{A3}
 \frac{\partial Q(\beta)}{\partial \beta_j} = \frac{\partial L(\beta|x^n)}{\partial \beta_j} + n\lambda_j f_j'(\beta_j) .
\end{equation}

Then by Taylor expanding at $\beta_0 = \left(\begin{matrix} \beta_{10} \\0 \end{matrix}\right) $ we have 
\begin{align*}
 \frac{\partial Q(\beta)}{\partial \beta_j} & = \frac{\partial L(\beta_0|x^n)}{\partial \beta_j} + \sum^p_{\ell =1} \frac{\partial^2 L(\beta_0|x^n)}{\partial \beta_j \partial \beta_{\ell}}(\beta_{\ell} - \beta_{\ell 0}) +  \sum^p_{\ell = 1}  \sum^p_{k =1} \frac{\partial^3 L(\beta^* | x^n)}{\partial \beta_j \partial \beta_{\ell} \partial \beta_k}(\beta_{\ell} - \beta_{\ell 0})(\beta_k - \beta_{k0})  \\
 &+ n \lambda_j f_j'(\beta_j)  
\end{align*}
where $\beta^*$ lies between $\hat{\beta}$ and $\beta_0$. Note that 
$$ 
\frac{1}{n} \frac{\partial L(\beta_0|x^n)}{\partial \beta_j} = O_p(n^{-\frac{1}{2}})
$$
and due to \cite{Hoadley:1971} and by the Law of Large Numbers for non-identically distributed random variables we have 
$$ 
\frac{1}{n} \frac{\partial^2 L(\beta_0|x^n)}{\partial \beta_j \partial \beta_{\ell}} =  E\left[ \frac{\partial^2 L(\beta_0|x^n)}{\partial \beta_{j} \partial \beta_{\ell}}  \right]+ o_p(1).
$$

Note that due to Conditions 1--6, and the result in Theorem 1, $\hat{\beta} - \beta_0 = O_p\left(\frac{1}{\sqrt{n}}\right)$.  So, we have

\begin{align*}
 \frac{\partial Q(\beta)}{\partial \beta_j} & =   n \frac{1}{n}\frac{\partial L(\beta_0|x^n)}{\partial \beta_j} + n\frac{1}{n} \sum^p_{\ell =1} \frac{\partial^2 L(\beta_0|x^n)}{\partial \beta_j \partial \beta_{\ell}}(\hat{\beta_{\ell}} - \beta_{\ell 0}) \\
 &+n\frac{1}{n}  \sum^p_{\ell = 1}  \sum^p_{k =1} \frac{\partial^3 L(\beta^* | x^n)}{\partial \beta_j \partial \beta_{\ell} \partial \beta_k}(\hat{\beta_{\ell}} - \beta_{\ell 0})(\hat{\beta_k} - \beta_{k0}) + n \lambda_j f_j'(\beta_j)  \\
 &=n  O_p\left(\frac{1}{\sqrt{n}}\right) + n  \sum^p_{\ell =1}\left( E\left[ \frac{\partial^2 L(\beta_0|x^n)}{\partial \beta_j \partial \beta_{\ell}}\right] +o_p(1)\right)O_p\left(\frac{1}{\sqrt{n}}\right)\\
&  + n \sum^p_{\ell = 1}  \sum^p_{k =1}E\left[ \sup_{\beta \in N} \left| \frac{ \partial^3 L(\beta^* | x^n)}{\partial \beta_j \partial \beta_{\ell} \partial \beta_k} \right| +o_p(1) \right]O_p\left(\frac{1}{\sqrt{n}}\right)O_p\left(\frac{1}{\sqrt{n}}\right) + n \lambda_j f_j'(\beta_j)\\
& = n \left[O_p\left(\frac{1}{\sqrt{n}}\right) + O_p\left(\frac{1}{\sqrt{n}}\right) + O_p\left(\frac{1}{\sqrt{n}}\right) \right] +  n \lambda_j f_j'(\beta_j)\\
& = n\left[ O_p\left(\frac{1}{\sqrt{n}}\right)+ \lambda_j f_j'(\beta_j) \right]\\
&= \sqrt{n}\left[ O_p(1) + \sqrt{n}\lambda_j f_j'(\beta_j)\right]\\
& \geq \sqrt{n}\left[ O_p(1) + \sqrt{n}b_n f_j'(\beta_j)\right]
\end{align*}
because $p$ is finite and using Slutsky's Theorem.   Since $\sqrt{n} b_n \rightarrow \infty$, 
the sign of $$
\frac{\partial Q(\beta)}{\partial \beta_j} = \sqrt{n}\left[ O_p(1) + \sqrt{n}\lambda_j f_j'(\beta_j)\right]
$$
is determined by the sign of $f_j'(\beta_j)$. 
Thus,  since $f_j'(0)=0$, we have 

$$
\begin{cases} 	
     f_j'(\beta_j) < 0 ,& \beta_j < 0,\\
     f_j'(\beta_j)=0 ,& \beta_j = 0 \\
     f_j'(\beta_j) > 0, & \beta_j > 0
   \end{cases}
$$
which implies that we have a minimum at $\beta_j=0$ and (\ref{min_beta}) is true.  
\end{proof}

The requirement that $b_n$ not go to zero too fast, if it goes to zero at all, is consistent with the fact that as the $\lambda_j$'s increase, the $\beta_j$'s must decrease to obtain an optimal solution.

The proof of the OP for penalized log-likelihood Theorem \ref{penllOP} is as follows.

\begin{proof}
First, we observe that under Conditions 1--6 and our assumptions on $b_n$,
Lemma \ref{sparsity} holds. So, $\hat{\beta}=(\hat{\beta_{1}}',\hat{\beta_2}')'$ satisfies $P(\hat{\beta_2}=0) \rightarrow 1$.

Now to prove the asymptotic normality of $\hat{\beta}_1$, consider $\frac{\partial Q(\beta)}{\partial \beta_j} $. It can be shown that there exists a $\hat{\beta_1}$ in Theorem 1 that is a 
$\sqrt{n}$-consistent minimizer of $Q \left\{\left(\begin{matrix} \beta_{1} \\0 \end{matrix} \right) \right\}$  which is a function of $\beta_1$ and satisfies the likelihood equation
 $$
\frac{\partial Q(\beta)}{\partial \beta_j}  \bigg|_{\hat{\beta}=\left(\begin{matrix} \hat{\beta_{1}} \\0 \end{matrix} \right)} = 0,
$$
for $j = 1, \ldots, p_0$. 
 We see that Taylor expanding at $\beta_0$ gives
\begin{align*}
0= \frac{\partial Q(\beta)}{\partial \beta_j}  \bigg|_{\hat{\beta}=\left(\begin{matrix} \hat{\beta_{1}} \\0 \end{matrix} \right)} &= \frac{\partial L(\hat{\beta}|x^n)}{\partial \beta_j} + n\frac{\partial}{\partial \beta_j} \sum^{p_0}_{j=1} \lambda_j f_j(\hat{\beta_j})\\
&=  \frac{\partial L({\beta_0} | x^n)}{\partial \beta_j}  + \sum^{p_0}_{j=1} \frac{\partial^2 {\L(\tilde{\beta} | x^n)}}{\partial \beta_{\ell}\partial \beta_j} (\hat{\beta_{\ell}} - \beta_{\ell 0}) + n\lambda_j f_j'(\hat{\beta_j}) 
\end{align*}
where $\tilde{\beta} \in <\beta_0 , \hat{\beta_1} >$. Furthermore, 
\begin{align}\label{asymp_normal}
\nonumber - \frac{\partial L(\beta_0 | x^n)}{\partial \beta_j} &= \sum^{p_0}_{j=1} \frac{\partial^2 L(\tilde{\beta} | x^n)}{\partial \beta_{\ell}\partial \beta_j} (\hat{\beta_{\ell}} - \beta_{\ell 0}) + n\lambda_j f_j'(\hat{\beta_j}) \\
\nonumber- \frac{\partial L(\beta_0 | x^n)}{\partial \beta_j}  &= n  \sum^{p_0}_{j=1} \frac{1}{n}\frac{\partial^2 L(\tilde{\beta} | x^n)}{\partial \beta_{\ell}\partial \beta_j} (\hat{\beta_{\ell}} - \beta_{\ell 0}) + n\lambda_j f_j'(\hat{\beta_j})\\
\nonumber- \frac{\partial L(\beta_0 | x^n)}{\partial \beta_j} & = - n \left(\sum^{p_0}_{j=1} I_{\ell j}(\beta_{10} | x^n) \right)(\hat{\beta_{\ell}} - \beta_{\ell 0}) + n\lambda_j f_j'(\hat{\beta_j})\\
\nonumber- \frac{\partial L(\beta_0 | x^n)}{\partial \beta_j}  & =- \sqrt{n} \left[ \sqrt{n}I_{1}(\beta_{10}| x^n) (\hat{\beta_{1}} - \beta_{10}) + \sqrt{n}\lambda_j f_j'(\hat{\beta_j})\right]\\
 \frac{1}{\sqrt{n}} \frac{\partial L(\beta_0 | x^n)}{\partial \beta_j}  & =  \sqrt{n}I_{1}(\beta_{10}| x^n) (\hat{\beta_{1}} - \beta_{10}) - \sqrt{n}\lambda_j f_j'(\hat{\beta_j})
\end{align}
where the third to last line is due to Condition 2.  Here $\hat{I}_1$ is like any $\hat{I}_n$ but restricted the upper block of $\hat{I}$.

To finish the proof, we examine the terms in (\ref{asymp_normal}). By supposition, $\sqrt{n}a_n \rightarrow 0$ and $f_j'(\hat{\beta_j})$ is uniformly bounded by Condition 4, so $\sqrt{n}\lambda_j f_j'(\hat{\beta_j}) \rightarrow 0$ because $\lambda_j < a_n$ for $p_0 < j \leq p$. Then, by the Central Limit Theorem, Condition 3 and the fact that $E\left(\frac{\partial L(\beta_0 | x^n)}{\partial \beta_j}\right)=0$, we have  
$$
\sqrt{n}\left(\frac{1}{n}\frac{\partial L(\beta_0 | x^n)}{\partial \beta_j} - E\left(\frac{\partial L(\beta_0 | x^n)}{\partial \beta_j}\right)\right)\overset {D} {\longrightarrow} \mathcal{N}(0 , I_1(\beta_{10})). 
$$
Thus, 
$$ \frac{1}{\sqrt{n}}\frac{\partial L(\beta_0 | x^n)}{\partial \beta_j} \overset {D} {\longrightarrow}  \mathcal{N}(0 , I_1(\beta_{10})).
$$
By equality in (\ref{asymp_normal}), it follows that 
$$
\sqrt{n} I_{1}(\beta_{10}| x^n) (\hat{\beta_{1}} - \beta_{10})  \overset {D} {\longrightarrow}   \mathcal{N}(0 , I_1(\beta_{10})). 
$$ 
\end{proof}

\subsection{Penalized Empirical Risk Case, Sec.  \ref{penemprisk}}
\label{penrisktheory}

\begin{thm}
\label{consistency2}
 Suppose Conditions 7-10 and Lemma \ref{exp_dprime} are satisfied.  If $\sqrt{n}a_n \rightarrow 0$ where $a_n$ satisfies s $a_n = \frac{1}{h(n)\sqrt{n}}$ where $\frac{1}{h(n)\sqrt{n}} \rightarrow 0$ and $h(n) \rightarrow \infty$ , then there exists a local minimizer $\hat{\beta}$ of $Q(\beta)$ such that $|| \hat{\beta} - \beta_0|| = O_p(n^{-\frac{1}{2}} + a_n)$.
\end{thm}

\begin{proof}
\textbf{Step 1.}
We want to show for any $\epsilon > 0$ there exists a large constant $C$ such that 
\begin{equation}\label{showthis2} P\left\{ \inf_{||u|| = C} Q(\beta_0 +\alpha_n u ) > Q(\beta_0)\right \}\geq 1-\epsilon .\end{equation}

Denote
\begin{align*}
D_n(u) & = Q(\beta_0 +\alpha_n u ) - Q(\beta_0) \\
 & = R(\beta_0 +\alpha_n u | x^n) + n \sum_{j=1}^p \lambda_j f( \beta_{j0} + \alpha_n u_j) - R(\beta_0| x^n) - n  \sum_{j=1}^p \lambda_j( f(\beta_{j0}))\\
 & = R(\beta_0 +\alpha_n u | x^n) - R(\beta_0| x^n) +  n  \sum_{j=1}^p \lambda_j [f_j( \beta_{j0} +\alpha_n u_j) - f_j( \beta_{j0})]  \\
  & = R(\beta_0 +\alpha_n u | x^n) - R(\beta_0| x^n) +  n \sum_{j=1}^{p} \lambda_j   f_j'(\tilde{\beta_j}) \alpha_n u_j\\
& = R(\beta_0 +\alpha_n u | x^n) - R(\beta_0| x^n) +   n \sum_{j=1}^{p_0} \lambda_j f_j'(\tilde{\beta_j}) \alpha_n u_j +  n  \sum_{j=p_0 +1}^{p} \lambda_j f_j'(\tilde{\beta_j}) \alpha_n u_j  
 \end{align*}
where $ f_j'(\tilde{\beta_j}) \alpha_n u_j$  is the Taylor expansion of the penalty term by 
Condition 10 and where $\tilde{\beta_j}$ is on the line joining $\beta_j$ to 
$\beta_j + \alpha_n u$. 
%More formally, $\tilde{\beta_j} \in <\beta_j , \beta_j +\alpha_n u > $.  
Thus we have that
the last two terms on the right are
the Taylor expansion of the penalty term by Condition 6.
%where $\tilde{\beta_j}$ is on the line joining $\beta_j$ to $\beta_j + \alpha_n u$. More formally, 
Thus,
\begin{align}
\nonumber D_n(u) & = R(\beta_0 +\alpha_n u | x^n) - R(\beta_0| x^n) +  n \alpha_n \left(  \sum_{j=1}^{p_0} \lambda_j f_j'(\tilde{\beta_j})  u_j  +n  \sum_{j=p_0 +1}^{p} \lambda_j f_j'(\tilde{\beta_j})  u_j \right)\\
 \nonumber  &= R(\beta_0 +\alpha_n u | x^n) - R(\beta_0| x^n) +  n (n^{-\frac{1}{2}} + a_n ) \left(  \sum_{j=1}^{p_0} \lambda_j f_j'(\tilde{\beta_j})  u_j  +n  \sum_{j=p_0 +1}^{p} \lambda_j f_j'(\tilde{\beta_j}) u_j \right) \\
 \nonumber &= R(\beta_0 +\alpha_n u | x^n) - R(\beta_0| x^n) + ( \sqrt{n}  + n a_n ) \left(  \sum_{j=1}^{p_0} \lambda_j f_j'(\tilde{\beta_j})  u_j  +n  \sum_{j=p_0 +1}^{p} \lambda_j f_j'(\tilde{\beta_j}) u_j \right)\\
 \nonumber & \geq R(\beta_0 +\alpha_n u | x^n) - R(\beta_0| x^n) + ( \sqrt{n}  + n a_n )a_n \left(  \sum_{j=1}^{p_0} f_j'(\tilde{\beta_j})  u_j  +n \sum_{j=p_0 +1}^{p} f_j'(\tilde{\beta_j})  u_j \right) \\
 \nonumber & = R(\beta_0 +\alpha_n u | x^n) - R(\beta_0| x^n) + ( \sqrt{n} a_n  + n a_n^2 ) \left( \sum_{j=1}^{p_0} f_j'(\tilde{\beta_j})  u_j  +n  \sum_{j=p_0 +1}^{p} f_j'(\tilde{\beta_j}) u_j \right) \\
 &= R(\beta_0 +\alpha_n u | x^n) - R(\beta_0| x^n) + \sqrt{n} a_n(1  + \sqrt{n} a_n ) \left(  \sum_{j=1}^{p_0}f_j'(\tilde{\beta_j})  u_j  +n  \sum_{j=p_0 +1}^{p} f_j'(\tilde{\beta_j}) u_j \right)  \label{dnu2}
 \end{align}
where $p_0$ is the number of components in $\beta_{10}$.

Now by Taylor expansion of the empirical risk at $\beta_0$, 
\begin{align}
\nonumber R(\beta_0 + \alpha_nu | x^n ) &- R(\beta_0| x^n)  = \alpha_n u R'(\beta_0 | x^n) + \frac{n}{2} (\alpha_n u)' \left(I^*(\beta_0 | x^n)\right) (\alpha_n u) \\
\nonumber & +  \frac{n}{2} (\alpha_n u)'\left(R''(\tilde{\beta}|x^n) \pm \frac{1}{n}\sum^n_{i=1} I^*(\beta| x^n) \pm I^*(\beta)\right) (\alpha_n u) \\ 
&= \alpha_n R'(\beta_0 | x^n) u  + \frac{n}{2} u' I^*(\beta_0 ) u  \alpha^2_n \{1 +o_p(1) \} \label{lnu2}
\end{align}
because  by Condition 7 we have 
$$ 
E\left[\sup_{\beta \in B(\beta_0,\eta)} \left|R''(\tilde{\beta}| x^n) - I^*(\beta_0 | x^n)\right| \right]\rightarrow 0
$$ 
as $\eta \rightarrow 0$ and where $I^*(\beta_0 | x^n) = E_L (R''_{jk} (\beta_0 | x^n ))_{jk}$.

Using $(\ref{dnu2})$ and $(\ref{lnu2})$ , 
\begin{eqnarray} 
D_n(u) &\geq&  \alpha_n R'(\beta_0 | x^n) u  + \frac{n}{2} u' I^*(\beta_0 | x^n) u  \alpha^2_n \{1 +o_p(1) \}
\nonumber \\
&+& \sqrt{n} a_n  ( 1  + \sqrt{n} a_n)   \sum_{j=1}^{p_0}f_j'(\tilde{\beta_j})  u_j 
\nonumber \\
&+& \sqrt{n} a_n  ( 1  + \sqrt{n} a_n)   \sum_{j=p_0+1}^{p}f_j'(\tilde{\beta_j})  u_j .
\label{A22}
\end{eqnarray}

Now we argue that the second term on the RHS of (\ref{A22}) dominates the others.  To see this, consider the first term on the RHS of (\ref{A22}). We multiply by $\frac{\sqrt{n}}{\sqrt{n}}$ and we have $\frac{\sqrt{n}}{\sqrt{n}} \alpha_n R'(\beta_0 | x^n) u$. Let $\bar{R} = \sqrt{n} \frac{\sum^n_{i=1} d'(y_i - x_i'\beta)}{n}$. Thus we rewrite $\frac{\sqrt{n}}{\sqrt{n}} \alpha_n R'(\beta_0 | x^n) u$ as $\frac{\alpha_n }{\sqrt{n}} \bar{R} u$. Now by the Central Limit Theorem, $$\sqrt{n} \left[ \frac{\sum^n_{i=1} d'(y_i - x_i'\beta)}{n} - E [ d'(y_i - x_i'\beta)]\right] $$ converges in distribution to a normal distribution since $E [ d'(y_i - x_i'\beta)]=0$ due to Lemma \ref{exp_dprime}  and because $\frac{ \alpha_n}{\sqrt{n}} \rightarrow 1$. Thus,  $C^2 = ||u||^2 = u'u > \bar{R} u$ as long as $0 < \frac{1}{B} < \bar{R} <  B$ for $C > B$ with high probability.  Then with high probability  $C^2 \geq B \underline{1}(u) \geq B \underline{1} u $ for $||u|| =C$.  Thus we have that by choosing sufficiently large $C$, the second term of (\ref{A2}) is larger in absolute value, with high probability than the first term uniformly in $||u|| =C$.

Also, by hypothesis $\sqrt{n} a_n \rightarrow 1 $, so the whole third term in (\ref{A22}) goes to zero and the second term of (\ref{A22}) is also larger than the third term. The fourth term of (\ref{A22}) also goes to zero because for $p_0 \leq j \leq p$, we know that $\beta_j =0$, so $\tilde{\beta_j} \rightarrow 0$ as $ n \rightarrow 0 $.Therefore  $\sum_{j=p_0+1}^{p}f_j'(\tilde{\beta_j})  u_j \rightarrow 0$ by Condition 9. It follows that the third term of (\ref{A22}) is larger than the fourth term, which implies that the second term of (\ref{A22}) is larger than all of the other terms. 

Therefore,  by choosing large enough $C$, 
$$ P\left\{\inf_{||u||=C} D_n(u) \geq  \frac{n}{2} u' I^*(\beta_0 | x^n) u  \alpha^2_n \{1 +o_p(1) \} \right \} \geq 1-\epsilon .$$ Thus, (\ref{showthis2}) is true for sufficiently large $C$.

\textbf{Step 2.}
We now show the conclusion of Step 1 holds for any $C^* > C$.
Let $C^\star > C$. Then,
\begin{align*}
\hat{\beta} & \in \{\beta_0 + \alpha_n u : ||u|| =C^* \} \\
&\iff \hat{\beta}  \in \cal{B}\left(\beta_0 , \alpha_n C^*\right) \\
&\iff \hat{\beta} - \beta_0  \in \cal{B}\left(0, \alpha_n C^* \right) \\
&\iff \alpha_n^{-1}(\hat{\beta} - \beta_0)  \in \cal{B}(0, C^*) \\
&\iff ||\alpha_n^{-1}(\hat{\beta} - \beta_0)|| < C^*\\
&\iff||\hat{\beta} - \beta_0||^2 < \alpha_n^2{C^*}^2
\end{align*}

It follows that  $ \sum^{p}_{j=1}(\hat{\beta}_j - \beta_{0j})^2 \leq \alpha_n^2{C^*}^2$  and $\forall j$ , $(\hat{\beta}_j - \beta_{0j})^2 \leq \alpha_n^2{C^*}^2$.  So,
$$ 
-\alpha_n C^* \leq \hat{\beta}_j - \beta_{0j} \leq \alpha_n C^*
$$ 
and we can absorb $C^*$ into $\alpha_n$ because $C^*$ is a constant. Thus we have $\hat{\beta} - \beta_0  = O_p(\alpha_n )$. Note that $$\alpha_n = \frac{1}{\sqrt{n}} + a_n = \frac{1+\sqrt{n}a_n}{\sqrt{n}}$$ and as $n \rightarrow \infty$ we have $\sqrt{n} a_n \rightarrow 0$, so $\hat{\beta} - \beta_0  = O_p\left(\frac{1}{\sqrt{n}}\right)$. 
\end{proof}

Again, we discuss the way you can get this result for penalty functions that have finitely many points that are non-differentiable, particularly at the origin in the following remark.

\begin{remark}
 The argument for Theorem 1 is true for differentiable $f_j(\cdot)$, but not non-differentiable $f_j(\cdot)$. In the case of non-differentiable $f_j(\cdot)$, consider a smooth function that approximates $f_j(\cdot)$ well. Say, $f_j^{\star}(\cdot)$, which differs from $f_j(\cdot)$ by a margin of $\epsilon$ where $\epsilon_N \rightarrow 0$ as $n \rightarrow \infty$. As long as $ \lim_{ \beta \rightarrow \beta^{* -}} f_j^{\star}(\cdot) =  \lim_{ \beta \rightarrow \beta^{* +}} f_j^{\star}(\cdot)$, then the above argument holds. 
\end{remark}

\begin{lemma}
\label{sparse}
Assume conditions 7--10 hold.  If $\sqrt{n}b_n \rightarrow \infty$ as $n \rightarrow \infty$ and and $b_n$ satisfies $b_n = \frac{g(n)}{\sqrt{n}}$ where $\frac{g(n)}{\sqrt{n}} \rightarrow 0$ and $g(n) \rightarrow \infty$, then with probability tending to 1, for any given $\beta_1$ satisfying $|| \beta_1 - \beta_{10}|| = O_p(n^{-\frac{1}{2}})$ and any constant $C$, 
\begin{equation}\label{beta_min2}
 Q \left \{ \left(\begin{matrix} \beta_1 \\0\end{matrix}\right) \right \} =  \max_{||\beta_2 ||\leq Cn^{-\frac{1}{2}}} Q \left \{ \left(\begin{matrix} \beta_1 \\ \beta_2\end{matrix}\right) \right \}.
 \end{equation}
\end{lemma}

\begin{proof}
Consider the objective function $Q(\beta) = R(\beta | x^n) + n \sum^p_{j=1} \lambda_j f_j(\beta_{j})$. Note that 
\begin{equation}\label{A3}
 \frac{\partial Q(\beta)}{\partial \beta_j} = \frac{\partial R(\beta | x^n)}{\partial \beta_j} + n\lambda_j f_j'(\beta_j) .
\end{equation}

Then by Taylor expanding at $\beta_0 = \left(\begin{matrix} \beta_{10} \\0 \end{matrix}\right) $ we have 
\begin{align*}
 \frac{\partial Q(\beta)}{\partial \beta_j}  = \frac{\partial R(\beta_0| x^n)}{\partial \beta_j} &+  \sum^p_{\ell =1} \frac{\partial^2 R(\beta_0| x^n)}{\partial \beta_j \partial \beta_{\ell}}(\beta_{\ell} - \beta_{\ell 0}) \\
 &+  \sum^p_{\ell = 1}  \sum^p_{k =1} \frac{\partial^3 R(\beta^* | x^n)}{\partial \beta_j \partial \beta_{\ell} \partial \beta_k}(\beta_{\ell} - \beta_{\ell 0})(\beta_k - \beta_{k0})  + n \lambda_j f_j'(\beta_j)  
\end{align*}

where $\beta^*$ lies between $\hat{\beta}$ and $\beta_0$. Note that $$ \frac{\partial R(\beta_0| x^n)}{\partial \beta_j} = O_p(n^{-\frac{1}{2}})$$
and due to \cite{Hoadley:1971} and by the Law of Large Numbers for non-identically distributed random variables we have $$   \frac{\partial^2 R(\beta_0| x^n)}{\partial \beta_j \partial \beta_{\ell}} =  E\left[ \frac{\partial^2 R(\beta_0| x^n)}{\partial \beta_{j} \partial \beta_{\ell}}  \right]+ o_p(1).$$

Also, due to Theorem 1, $\hat{\beta} - \beta = O_p\left(\frac{1}{\sqrt{n}}\right)$, so we have

\begin{align*}
 \frac{\partial Q(\beta)}{\partial \beta_j} & =    \frac{\partial R(\beta_0| x^n)}{\partial \beta_j} +  \sum^p_{\ell =1} \frac{\partial^2 R(\beta_0| x^n)}{\partial \beta_j \partial \beta_{\ell}}(\hat{\beta_{\ell}} - \beta_{\ell 0}) \\
 &+  \sum^p_{\ell = 1}  \sum^p_{k =1} \frac{\partial^3 R(\beta^*|x^n)}{\partial \beta_j \partial \beta_{\ell} \partial \beta_k}(\hat{\beta_{\ell}} - \beta_{\ell 0})(\hat{\beta_k} - \beta_{k0})  + n \lambda_j f_j'(\beta_j)  \\
 &= O_p\left(\frac{1}{\sqrt{n}}\right) +   \sum^p_{\ell =1}\left( E\left[ \frac{\partial^2 R(\beta_0| x^n)}{\partial \beta_j \partial \beta_{\ell}}\right] +o_p(1)\right)O_p\left(\frac{1}{\sqrt{n}}\right)\\
&  +  \sum^p_{\ell = 1}  \sum^p_{k =1}E\left[ \sup_{\beta \in N} \left| \frac{ \partial^3 R(\beta^* | x^n)}{\partial \beta_j \partial \beta_{\ell} \partial \beta_k} \right| +o_p(1) \right]O_p\left(\frac{1}{\sqrt{n}}\right)O_p\left(\frac{1}{\sqrt{n}}\right) + n \lambda_j f_j'(\beta_j)\\
& =  \left[O_p\left(\frac{1}{\sqrt{n}}\right) + O_p\left(\frac{1}{\sqrt{n}}\right) + O_p\left(\frac{1}{\sqrt{n}}\right) \right] +  n \lambda_j f_j'(\beta_j)\\
& =  O_p\left(\frac{1}{\sqrt{n}}\right)+ n\lambda_j f_j'(\beta_j) \\
&= \sqrt{n}\left[ O_p\left(\frac{1}{n}\right) + \sqrt{n}\lambda_j f_j'(\beta_j)\right]\\
& \geq \sqrt{n}\left[ O_p\left(\frac{1}{n}\right) + \sqrt{n}b_n f_j'(\beta_j)\right]
\end{align*}
because $p$ is finite and using Slutsky's Theorem.  Since $\sqrt{n} b_n \rightarrow \infty$,  the sign of 
$$
\frac{\partial Q(\beta)}{\partial \beta_j} = \sqrt{n}\left[ O_p\left(\frac{1}{n}\right) + \sqrt{n}\lambda_j f_j'(\beta_j)\right]
$$ 
is determined by the sign of $f_j'(\beta_j)$.  That is,  since $f_j'(0)=0$, we have 
$$
\begin{cases} 	
     f_j'(\beta_j) < 0 ,& \beta_j < 0,\\
     f_j'(\beta_j)=0 ,& \beta_j = 0 \\
     f_j'(\beta_j) > 0, & \beta_j > 0
   \end{cases}
$$
which implies that we have a minimum at $\beta_j=0$ and (\ref{beta_min2}) is true.  
\end{proof}

The proof of the OP for penalized log-likelihood Theorem \ref{empriskOP} is as follows.

\begin{proof}
First, we observe that the estimator $\hat{\beta}=(\hat{\beta_{1}}',\hat{\beta_2}')'$ satisfies $P(\hat{\beta_2}=0) \rightarrow 1$ due to Lemma \ref{sparse}.

Now to prove the asymptotic normality piece, consider $\frac{\partial Q(\beta)}{\partial \beta_j} $. It can be shown that there exists a $\hat{\beta_1}$ in Theorem 3 that is a $\sqrt{n}$ consistent minimizer of $Q \left\{\left(\begin{matrix} \beta_{1} \\0 \end{matrix} \right) \right\}$  which is a function of $\beta_1$ and satisfies the likelihood equations 
$$
\frac{\partial Q(\beta)}{\partial \beta_j}  \bigg|_{\hat{\beta}=\left(\begin{matrix} \hat{\beta_{1}} \\0 \end{matrix} \right)} = 0 
$$ 
for $j = 1, \ldots, p_0$. 

 We see that Taylor expanding at $\beta_0 $ gives
\begin{align*}
0= \frac{\partial Q(\beta)}{\partial \beta_j}\bigg|_{\hat{\beta}=\left(\begin{matrix} \hat{\beta_{1}} \\0 \end{matrix} \right)} &= \frac{\partial R(\hat{\beta} | x^n)}{\partial \beta_j}   + \frac{\partial}{\partial \beta_j} n \sum^{p_0}_{j=1} \lambda_j f_j(\hat{\beta_j})\\
&=  \frac{\partial R(\beta_0 | x^n)}{\partial \beta_j}  + \sum^{p_0}_{j=1} \frac{\partial^2 R(\tilde{\beta} | x^n)}{\partial \beta_{\ell}\partial \beta_j} (\hat{\beta_{\ell}} - \beta_{\ell 0}) + n\lambda_j f_j'(\hat{\beta_j}) \\
\end{align*}
where $\tilde{\beta} \in <\beta_0 , \hat{\beta_1} >$. Furthermore, 

\begin{align}\label{asymp_normal2}
\nonumber - \frac{\partial R(\beta_0 | x^n)}{\partial \beta_j} &= \sum^{p_0}_{j=1} \frac{\partial^2 R(\tilde{\beta} | x^n)}{\partial \beta_{\ell}\partial \beta_j} (\hat{\beta_{\ell}} - \beta_{\ell 0}) + n\lambda_j f_j'(\hat{\beta_j}) \\
\nonumber-\frac{\partial R(\beta_0 | x^n)}{\partial \beta_j}  &= n  \sum^{p_0}_{j=1} \frac{1}{n}\frac{\partial^2 R(\tilde{\beta} | x^n)}{\partial \beta_{\ell}\partial \beta_j} (\hat{\beta_{\ell}} - \beta_{\ell 0}) + n\lambda_j f_j'(\hat{\beta_j})\\
\nonumber- \frac{\partial R(\beta_0 | x^n)}{\partial \beta_j}  & = -n \sum^{p_0}_{j=1} I^*_{\ell j}(\beta_{10} | x^n) (\hat{\beta_{\ell}} - \beta_{\ell 0}) + n\lambda_j f_j'(\hat{\beta_j})\\
\nonumber- \frac{\partial R(\beta_0 | x^n)}{\partial \beta_j}  & = - \sqrt{n} \left[ \sqrt{n}I^*_{1}(\beta_{10}| x^n) (\hat{\beta_{1}} - \beta_{10}) + \sqrt{n}\lambda_j f_j'(\hat{\beta_j})\right]\\
\frac{1}{\sqrt{n}}  \frac{\partial R(\beta_0 | x^n)}{\partial \beta_j}  & =  \sqrt{n}I^*_{1}(\beta_{10}| x^n) (\hat{\beta_{1}} - \beta_{10}) - \sqrt{n}\lambda_j f_j'(\hat{\beta_j})
\end{align}
where the third to last line is due Condition 7.  Here $I^*_1$ is like any $I^*_n$ but restricted the upper block of $I^*$.

To finish the proof, we examine the terms in (\ref{asymp_normal2}).  By supposition, $\sqrt{n}a_n \rightarrow 0$ and $f_j'(\hat{\beta})$ is bounded. Thus, $\sqrt{n}\lambda_j f_j'(\hat{\beta_j}) \rightarrow 0$ for $p_0 < j \leq p$  since $\lambda_j < a_n$. Then by the Central Limit Theorem, Condition 3 and using Lemma \ref{exp_dprime} and the linearity of the expectation operator we know $E\left(\frac{\partial R(\beta_0 | x^n)}{\partial \beta_j}\right)=0$, so we have  
$$
\sqrt{n}\left(\frac{1}{n}\frac{\partial R(\beta_0 | x^n)}{\partial \beta_j} - E\left(\frac{\partial R(\beta_0 | x^n)}{\partial \beta_j}\right)\right)\overset {D} {\longrightarrow}  \mathcal{N}(0 , I^*_1(\beta_{10})). 
$$
Thus, 
$$
 \frac{1}{\sqrt{n}}\frac{\partial R(\beta_0 | x^n)}{\partial \beta_j} \overset {D} {\longrightarrow}  \mathcal{N}(0 , I^*_1(\beta_{10})).
$$ 
By equality in (\ref{asymp_normal}), it follows that 
$$ 
\sqrt{n} I^*_{1}(\beta_{10}| x^n) (\hat{\beta_{1}} - \beta_{10})  \overset {D} {\longrightarrow} \mathcal{N}(0 , I^*_1(\beta_{10})).  
$$ 
\end{proof}

%\end{supplement}

\section*{References}

\bibliographystyle{elsarticle-num}
\bibliography{shrinkagerefs}

\begin{thebibliography}{10}
\expandafter\ifx\csname url\endcsname\relax
  \def\url#1{\texttt{#1}}\fi
\expandafter\ifx\csname urlprefix\endcsname\relax\def\urlprefix{URL }\fi
\expandafter\ifx\csname href\endcsname\relax
  \def\href#1#2{#2} \def\path#1{#1}\fi

\bibitem{Hoerl:1962}
A.~Hoerl, Application of ridge analysis to regression problems, Chemical
  Engineering Progress 58 (1962) 54--59.

\bibitem{Sjoburg:Ljung:1992}
J.~Sj{\"o}burg, L.~Ljung, Overtraining, regularization, and searching for
  minimum in neural networks, in: Proceedings of the 4th IFAC Symposium on
  Adaptive Systems in Control and Signal Processing, Vol.~25, 1992, pp. 73--78.

\bibitem{Tibshirani:1996}
R.~Tibshirani, Regression shrinkage and selection via the lasso, J. Roy.
  Statist. Soc., Ser. B 58 (1996) 267--288.

\bibitem{Fan:Li:2001}
J.~Fan, R.~Li, Variable selection via concave penalized likelihood and its
  oracle properties, J. Amer. Statist. Assoc. 96 (2001) 1348--1360.

\bibitem{Zhang:2010}
C.-H. Zhang, Nearly unbiased variable selection under minimax concave penalty,
  Ann. Statist. 38 (2010) 894--942.

\bibitem{Zou:Hastie:2005}
H.~Zou, T.~Hastie, Regularization and variable selection via the elastic net,
  J. Roy. Stat. Soc., Ser. B. 67 (2005) 301–320.

\bibitem{Zou:Zhang:2009}
H.~Zou, H.~H. Zhang, On the adaptive elastic-net with a diverging number of
  parameters, Ann. Statist. 37 (2009) 1733 -- 1751.

\bibitem{Zou:2006}
H.~Zou, The adaptive lasso and its oracle properties, J. Amer. Statist. Assoc.
  101 (2006) 1418--1429.

\bibitem{Wang:etal:2007}
H.~Wang, G.~Li, G.~Giang, Robust regression shrinkage and consistent variable
  selection through the lad-lasso., J. business and Economic Stat. 25 (2007)
  347--355.

\bibitem{Gian:Wang:2013}
W.~Gian, Y.~Yang, Model selection via standard error adjusted adaptive lasso,
  Ann. Inst. Stat. Math 65 (2013) 295--318.

\bibitem{Fan:Lv:2013}
J.~Fan, J.~Lv, Asymptotic equivalence of regularization methods in thresholded
  parameter spaces, J. Amer. Statist. Assoc. 108 (2013) 1044--1061.

\bibitem{Luo:etal:2006}
X.~Luo, L.~Stefanski, D.~Boos, Variable selection via concave penalized
  likelihood and its oracle properties, Technometrics 48 (2006) 165--175.

\bibitem{Friedman:etal:2010}
J.~Friedman, T.~Hastie, R.~Tibshirani,
  \href{https://www.jstatsoft.org/v33/i01/}{Regularization paths for
  generalized linear models via coordinate descent}, Journal of Statistical
  Software 33~(1) (2010) 1--22.
\newline\urlprefix\url{https://www.jstatsoft.org/v33/i01/}

\bibitem{Sherwood:etal:2020}
B.~Sherwood, A.~Maidman, \href{https://CRAN.R-project.org/package=rqPen}{rqPen:
  Penalized Quantile Regression}, r package version 2.2.2 (2020).
\newline\urlprefix\url{https://CRAN.R-project.org/package=rqPen}

\bibitem{Fan:etal:2014}
J.~Fan, L.~Xue, H.~Zou, Strong oracle optimality of folded concave penalized
  estimation, Ann. Stat. 42 (2014) 819--849.

\bibitem{Zhao:Yu:2006}
P.~Zhao, B.~Yu, On model selection consistency selection of lasso, J. M. L. R.
  7 (2006) 2541--2563.

\bibitem{Hamidieh:2018}
K.~Hamidieh, A data-driven statistical model for predicting the critical
  temperature of a superconductor, J. Amer. Statist. Assoc. 96 (2018)
  1348--1360.

\bibitem{Givens:Hoeting:2013}
G.~Givens, J.~Hoeting, Computational Statistics, 2nd Edition, Wiley, Hoboken,
  NJ, 2013.

\bibitem{Rudolph:1996}
G.~Rudolph, Convergence of evolutionary algorithms in general search spaces,
  in: Proc. of IEEE International Conference on Evolutionary Computation, 1996,
  pp. 50--54.
\newblock \href {http://dx.doi.org/10.1109/ICEC.1996.542332}
  {\path{doi:10.1109/ICEC.1996.542332}}.

\bibitem{Willighagen:etal:2015}
E.~Willighagen, M.~Ballings,
  \href{https://CRAN.R-project.org/package=genalg}{genalg: R Based Genetic
  Algorithm}, r package version 0.2.0 (2015).
\newline\urlprefix\url{https://CRAN.R-project.org/package=genalg}

\bibitem{Royden:FitzPatrick:2010}
H.~Royden, P.~FitzPatrick, Real Analysis, 4th Edition, Prentice-Jall, Hoboken,
  NJ, 2010.

\bibitem{Hoadley:1971}
B.~Hoadley, Asymptotic properties of maximum likelihood estimators for the
  independent not identically distributed case., Ann. math. Statsit. 40 (1971)
  1977--1991.

\end{thebibliography}

%\begin{acks}[Acknowledgments]
%And this is an acknowledgements section with a heading that was produced by the
%$\backslash$section* command. Thank you all for helping me writing this
%\LaTeX\ sample file.
%\end{acks}

\end{document}